\documentclass[11pt]{article}
\usepackage[utf8]{inputenc}
\usepackage{amsmath, amssymb, amsthm}
\usepackage{graphicx}
\usepackage{geometry}
\usepackage{authblk}
\usepackage{xr-hyper}
\usepackage{hyperref}
\usepackage{algorithm}
\usepackage[noend]{algpseudocode}
\usepackage{caption}
\usepackage{subcaption}
\usepackage{fancyhdr}
\usepackage{booktabs}
\usepackage{dsfont}
\usepackage{tikz}

\newcommand{\D}{\mathcal{D}}

\newcommand{\EE}{\mathbb{E}}
\newcommand{\RR}{\mathbb{R}}
\newcommand{\PP}{\mathbb{P}}

\DeclareMathOperator*{\argmax}{argmax}
\DeclareMathOperator*{\argmin}{argmin}

\newtheorem{prop}{Proposition}
\newtheorem{lema}{Lemma}
\newtheorem{teo}{Theorem}
\theoremstyle{definition}

\newtheorem{defin}{Definition}
\newtheorem{cor}[prop]{Corollary}

\begin{document}

\title{\huge{ Differentially Private Algorithms for Linear Queries via Stochastic Convex Optimization}}

\author{
  \textbf{Giorgio Micali}$^{1}$, \textbf{Clement Lezane}$^{1}$, and \textbf{Annika Betken}$^{1}$ \\
  \texttt{g.micali@utwente.nl, c.lezane@utwente.nl, a.betken@utwente.nl}
}
\date{}  % This removes the date
\maketitle
\footnotetext[1]{University of Twente, Faculty of Electrical Engineering, Mathematics, and Computer Science (EEMCS), Drienerlolaan 5, 7522 NB Enschede, Netherlands}

% ----------------- abstract - start ------------------------------

\begin{abstract}

This article establishes a method to answer a finite set of  linear queries on a given dataset while ensuring differential privacy.
To achieve this, we
formulate the corresponding task  as a saddle-point problem, i.e. an optimization problem  whose solution corresponds to a distribution minimizing the difference between answers to the linear  queries based on the true distribution and answers from a differentially private distribution.  
Against this background, we establish two new algorithms for corresponding differentially private data release:
the first is 
based on the differentially private Frank-Wolfe method,  the second combines randomized smoothing with stochastic convex optimization techniques for a solution to the saddle-point problem.
While previous works assess the accuracy of 
differentially private algorithms
with reference to
the empirical  data distribution, a key contribution of our work is a more  natural evaluation of the proposed algorithms' accuracy 
with reference to the true data-generating distribution. 
\end{abstract}

% ----------------- abstract - end ------------------------------

% ----------------- introduction - start ------------------------------

\section{Introduction}
As data analysis plays a pivotal role in modern society, the need to protect private information on individuals  becomes increasingly important. Ensuring privacy of sensitive information requires robust and effective methods. Among various approaches, differential privacy \cite{dwork2006calibrating}  emerged as a widely recognized and adopted framework. In fact, today, many public  data releases, such as the 2020 US Census \cite{example2}, are required to adhere to differential privacy standards.

This article studies private synthetic data generation for query release, i.e. 
it aims at the construction of synthetic datasets 
as replacement for datasets containing sensitive information, while  balancing  
the accuracy of  answers to a predefined set of data queries and the protection of sensitive information in the original dataset. 

Unlike commonly known standard privatization methods that are based on the principle of data protection obtained through the incorporation of noise at  output level, we propose algorithms that generate a synthetic data distribution 
that  is subject to differential privacy, while close to the data-generating distribution.
One fundamental advantage of this method is that once a differentially private distribution is obtained, a  synthetic, private dataset  can be generated by  sampling independently from this distribution.

\subsection*{Related Works}

Generation of  private synthetic data for query release is at the heart of research within the field of differential privacy. 
Many established approaches
aiming at (approximate) accuracy in answering a high number of statistical queries through synthetic data release, while, at the same time preserving privacy,  
are based on the Private Multiplicative Weights (PMW) mechanism introduced by \cite{hardt2010multiplicative}.
Starting  from a uniform distribution on the data universe, PMW iteratively updates a synthetic data distribution based on the errors of all previously answered queries.
%which can be computationally expensive. 
Recent works  improved the computational efficiency of PMW:
The Multiplicative Weights Exponential Mechanism (MWEM), introduced in  \cite{HM12}, builds upon PMW by using the exponential mechanism to focus on the queries with the largest errors, optimizing a surrogate loss function. This allows MWEM to achieve better accuracy and faster convergence, making it more efficient than PMW for practical purposes.
 While MWEM maintains a full distribution over the data domain, which leads to high computational complexity especially in high-dimensional settings, DualQuery (established in \cite{pmlr-v32-gaboardi14}) instead maintains a distribution over the query class, which is typically much smaller than the data domain. Moreover, \cite{vtb20}
 proposes three new algorithms for synthetic data release: FEM, sepFEM and DQRS. 
 FEM (Follow-the-Perturbed-Leader with Exponential Mechanism) selects queries using the Exponential Mechanism (like MWEM), but differs by updating through a perturbed optimization problem with stochastic updates. sepFEM is a variant of FEM which achieves better error rates for some classes of queries. The algorithm, however,  is less general as it relies on the assumption that for any two distinct pairs of elements in the data domain, there must exist a query for which their answers differ. DQRS is a variant of DualQuery, which implements rejection sampling to reuse previous samples, reducing the number of oracle calls and queries that had not been sampled before.
An overview of these algortihms' error bounds can be found in Table \ref{tab:synthetic_algorithms} (see also \cite{vtb20}). 

Another two methods for private, synthetic data release have been proposed by \cite{LVW21}:  Private Entropy Projection (PEP) and Generative Networks with the Exponential Mechanism (GEM). PEP builds upon MWEM by adaptively reusing past query measurements and optimizing learning rates resulting in improved accuracy and faster convergence than MWEM. It, however, retains the computational overhead of iterating over the full data distribution. GEM addresses this issue by training neural networks to compactly represent data distributions. While GEM shows strong empirical performance in high-dimensional settings on the one hand, it raises new challenges related to neural network training and hyperparameter tuning on the other.

Other relevant work in the context of differentially private query release include \cite{pmlr-v119-bassily20a} which studies the private and public sample complexities of the Private Query Release Assisted by Public Data (PAP) algorithms and gives upper and lower bounds on both. Another key contribution is \cite{pmlr-v139-aydore21a} which introduces the Relaxed Adaptive Projection (RAP) Mechanism, involving differentially private optimization techniques to address high-dimensional query release problems. Lastly, we highlight recent work by \cite{ZSJ24} which presents an alternative approach to private synthetic data generation for images through the use of foundation models.

\subsection*{Our contributions}
 We revisit the query release problem formulation introduced by \cite{LVW21}, where the goal is to find a distribution within a specified family that minimizes the maximum error across all queries. Unlike previous approaches that rely on the Adaptive Measurements algorithm \cite{hardt2010multiplicative}—which minimizes a surrogate loss at each iteration to solve an easier problem—we directly tackle the core objective by optimizing a regularized version of the maximum error among all query answers within the framework of Differentially Private Stochastic Convex Optimization (DP-SCO). Instead of employing adaptive measurements, we solve the problem using differentially private optimization techniques, proposing two algorithms: a private Frank-Wolfe method (DPFW)
 and a randomized smoothing mirror descent approach (DPAM). 
 By selecting an appropriate regularization constant, we
 minimize accuracy bounds for 
 both methods achieve.

Furthermore, we assess the accuracy of our methods DPFW and DPAM by providing guarantees with respect to the true data-generating distribution.%; see the last three rows of Table \ref{tab:synthetic_algorithms} in Section \ref{sec:conclusion}. 
To the best of our knowledge, this is the first work that offers such guarantees.
In fact, and in contrast to our analysis, prior works  typically analyze the  empirical risk only, i.e. the risk computed from the dataset’s empirical distribution.%; see the first five rows of Table \ref{tab:synthetic_algorithms} for corresponding accuracy guarantees. 

Notably, the accuracy of DPFW with respect to the population loss aligns with its accuracy on the empirical counterpart. The latter is comparable to most established state-of-the-art accuracy bounds for the aforementioned methods. Furthermore, the upper bound for the population loss of our second algorithm coincides with its empirical counterpart, and for a large query class, DPAM achieves higher accuracy compared than all previously discussed methods.

% -------------- RELATED WORKS end ---------

\section{Preliminaries}\label{sec:preliminaries}

In the following, we establish notations and assumptions that will be used consistently throughout this manuscript.
Let $\mathcal{Z}$, the data universe,  be a finite sample space and  $\mathcal{P}$ a probability distribution supported
on $\mathcal{Z}$. Without loss of generality we assume that $\mathcal{Z}=\{1, \ldots, k\}$, where $|\mathcal{Z}|=k$. We refer to $\mathcal{P}$ as  {\em target distribution}. Due to finiteness of $\mathcal{Z}$, $\mathcal{P}$ can be represented
as $(\mathcal{P}(z))_{z \in \mathcal{Z}}$. We define the set of probability distributions  on $\mathcal{Z}$ by
$\Delta_{k} := \{ (\mathcal{D}(z))_{z \in \mathcal{Z}} \in \mathbb{R}^{k}_+ : \sum_{z\in \mathcal{Z}} \mathcal{D}(z) = 1 \}$. Let us denote with $S_n = \{z_1, \dots, z_n\}$ a dataset of $n$ points obtained by sampling independently from $\mathcal{Z}$ according to  $\mathcal{P}$. From $S_n$ we retrieve the empirical distribution $\mathcal{P}_n$, i.e. $\mathcal{P}_n(z) = |\{i \in [n] \,:\,  z_i=z\}|/n$. We refer to any function 
$q : \mathcal{Z} \to [-1, 1]$ as a {\em query} of $\mathcal{Z}$, and we define the {\em statistical query} associated to $q$ with respect to the distribution $\mathcal{P}$ as the
quantity $ \mathbb{E}_{\mathbf{z} \sim \mathcal{P}}[q(\mathbf{z})]$. 
If $\mathcal{P}_n$ is the empirical distribution associated to $S_n $, the statistical query
takes the form $\frac{1}{n} \sum_{i} q(z_i)$, which is known as \textit{counting query} \cite{dwork2014algorithmic}. 
Since each query $q$ is a function defined on $\mathcal{Z}$, it can be represented as  $q = (q(z))_{z \in \mathcal{Z}} \in [-1, 1]^{k}$. 
A query is said to be {\em answered} when the vector $q = (q(z))_{z \in \mathcal{Z}}$ is reported.
Similarly, any distribution $\mathcal{D} \in \Delta_{k}$ can be expressed as $\mathcal{D} = (\mathcal{D}(z))_{z \in \mathcal{Z}}$. Depending on the context, we refer to queries $q$ and distributions $\mathcal{D}$ as functions or vectors. Against this background, $\mathbb{E}_{\mathbf{z} \sim \mathcal{P}}[q(\mathbf{z})]=\langle q, \mathcal{P}\rangle$, where $\langle \cdot, \cdot \rangle$ is the standard scalar product in $\RR^{k}$.

In many practical applications, queries take two values only. 
If, for instance,  $S_n$ represents a set of individuals, a query $q$ could 
ask for the presence/absence of a specific feature in these individuals. In this particular case, a corresponding query may only take the  values $0$ and $1$ ($q(z)=1$ indicating the presence of the feature for the individual $z$, while $q(z)=0$ indicates its absence).
The statistical query $ \mathbb{E}_{\mathbf{z} \sim \mathcal{P}}[q(\mathbf{z})]$ then corresponds to the expected proportion of individuals with a corresponding feature.

Within the field of differential privacy, an objective is to answer a prescribed set of queries $\mathcal{Q}\subset \{ q:\mathcal{Z}\longrightarrow [-1,1]\}$ as accurately as possible, while  preserving the privacy of  individuals. For this,  we construct a privacy-preserving distribution
$\mathcal{P}^{\text{priv}}$ over $\mathcal{Z}$, whose statistical properties are close to $\cal P$, i.e., more precisely,  
that its statistical queries closely
approximate those of  $\mathcal{P}$.
Mathematically speaking,
this is achieved by minimization of
$\max\limits_{q \in \mathcal{Q}} \left|\mathbb{E}_{\mathbf{z} \sim \mathcal{P}}[q(\mathbf{z})] - \mathbb{E}_{\mathbf{z} \sim \mathcal{P}^{\text{priv}}}[q(\mathbf{z})]\right|$, over all $\mathcal{P}^{\text{priv}}$  subject to privacy, in the sense of differential
privacy defined below. For this, we assume $\mathcal{Q}$ to be symmetric, i.e.  $\mathcal{Q}=-\mathcal{Q}$, such that, we could as well be aiming at minimizing 
$\max_{q \in \mathcal{Q}} \left(\mathbb{E}_{\mathbf{z} \sim \mathcal{P}}[q(\mathbf{z})] - \mathbb{E}_{\mathbf{z} \sim \mathcal{P}^{\text{priv}}}[q(\mathbf{z})]\right)$.

Since $\text{conv}(\mathcal{Q})$ is a polyhedron and  the maximum of a linear function over a polyhedron is achieved at one of its vertices, we can, furthermore,  replace $\mathcal{Q}$ by its convex hull. 
As a result, the considered optimization problem can be reduced to the   {\em saddle-point formulation} 
\begin{equation}
\label{primal}
\min\limits_{\mathcal{D} \in \Delta_{k}}\phi(\mathcal{D}), 
\ \text{where }
\ \phi(\mathcal{D}):=\max\limits_{q \in \mbox{conv}(\mathcal{Q})}\langle q, \mathcal{P}- \mathcal{D} \rangle \tag{\textcolor{blue}{$P$}}
\end{equation}
see \cite{LVW21}.
Our goal is to design a differentially private  algorithm that approximates the optimal solution to \eqref{primal} starting from a finite dataset $S_n = \{z_1, \dots, z_n\}$ consisting of $n$ data points independently drawn from $\mathcal{Z}$ according to $\mathcal{P}$.

\subsection*{Differential Privacy}

Differential privacy applies to randomized algorithms, i.e. algorithms  employing randomness as part of their logic or procedure
for the protection of  privacy.
 A randomized algorithm   is considered differentially private if adding, removing or replacing a single data point in the dataset it applies to does not significantly change the distribution of its output. 
Two datasets 
 $S_1$ and $S_2$ that
 differ in one data point only are called 
{\em neighboring}. 
%We write $S_1 \sim S_2$ to indicate that $S_1$ and $S_2$ are neighboring datasets.
 
Mathematically, differential privacy can be  formalized 
as follows:

\begin{defin} \cite{dwork2014algorithmic}
Let $\varepsilon, \delta \geq  0 $. A randomized algorithm $\mathcal{A}: \mathcal{X}\longrightarrow \mathcal{O}$ is  $(\varepsilon,\delta)$ - differentially  private (DP) if for every pair of neighboring datasets $S_1$, $S_2$ and every subset $T\subset \mathcal{O}$,
\begin{equation}
 \PP( \mathcal{A}(S_1) \in T) \leq e^\varepsilon \PP(\mathcal{A}(S_2) \in T)+\delta\;.
 \label{def:eps-delta-DP}
\end{equation}
\label{def:DP}
\end{defin}
Note that the parameters $\varepsilon$ and $\delta$ quantify  privacy. If $\varepsilon=0$ and $\delta=0$, it follows  that $\PP (\mathcal{A}(S_1)\in T)=\PP(\mathcal{A}(S_2)\in T)$. Thus, the distribution of the output of $\mathcal{A}$ is independent from the data and protects privacy perfectly. If $\delta=0$, we say that $\mathcal{A}$ satisfies pure privacy and/or that it is  $\varepsilon-$DP.

% ----------------- Synthetic Data - start ------------------------------

% ----------------- Synthetic Data - end ------------------------------

\section{A regularized optimization problem}\label{sec:regularization}

Without imposing any restrictions on the set of admissible distributions
$\Delta_{k}$ in the saddle-point problem \eqref{primal}, the solution to this optimization problem is given by  $ {\mathcal{D}}^\ast =\mathcal{P}$.
Our goal, however, is to find an approximation to this solution that is $(\varepsilon, \delta)$-DP.
Moreover, the optimization problem as stated in \eqref{primal} does not satisfy strong convexity. For this reason,   we modify the optimization problem through incorporation of a regularization transforming it to a strongly convex-concave problem.

More precisely, for $\alpha \geq 0$, we define the {\em regularized saddle-point problem }\eqref{primal_alpha} as
\begin{equation}
\begin{aligned}
    &\min\limits_{\mathcal{D} \in \Delta_{k}} \max_{q \in \text{conv}(\mathcal{Q})}\mathcal{L}_\alpha(q, \mathcal{D})\quad 
    \label{primal_alpha}
    =    &\min\limits_{\mathcal{D} \in \Delta_{k} }\phi_\alpha(\mathcal{D}),
\end{aligned}
\tag{$\textcolor{blue}{P_\alpha}$}
\end{equation}
where 
\begin{align*}
    \mathcal{L}_{\alpha}(q, \mathcal{D})&:= \langle q, \mathcal{P}-\mathcal{D} \rangle + \alpha H(\mathcal{D})\phi_{\alpha}(\mathcal{D}):= \max\limits_{q \in \text{conv}(\mathcal{Q})}\mathcal{L}_\alpha(q, \mathcal{D})
\end{align*}
and  $H(\mathcal{D})=\sum_{z\in \mathcal{Z}} \mathcal{D}(z) \log (\mathcal{D}(z))$ is the negative entropy of $\mathcal{D} \in\Delta_{k}$.

The use of the entropy regularization term is motivated by two key observations: first, in practice, the set $\mathcal{Z}$ often contains a number of elements that is much larger than the sample size, i.e. $n < k$. 
As a result, for the empirical distribution $\mathcal{P}_n$ over $\mathcal{Z}$,  $\mathcal{P}_n(z) = 0$ for some $z\in \mathcal{Z}$. Entropy regularization addresses this by pushing the optimum towards the uniform distribution, thereby ensuring that most points in $\mathcal{Z}$ have a non-zero probability of being sampled, even if they are not present in $S_n$. Secondly, the entropy function $H$ is strongly convex ensuring 
 strong convexity of $\phi_{\alpha}$  and therefore a unique solution of \eqref{primal_alpha}.

The regularization parameter $\alpha$ controls the suboptimality of the solution with respect to the unregularized problem.
A value of $\alpha$ that is close to zero ensures the regularized problem \eqref{primal_alpha} to be similar to the original problem \eqref{primal}. %After solving \eqref{primal_alpha}, we will assess how much the resulting solution deviates from the actual distribution $\D$. 

For an explicit solution of the regularized optimization problem \eqref{primal_alpha}, we consider the corresponding dual problem:
%The dual often provides necessary and sufficient conditions for optimality in the primal problem.   In fact, primal and dual  are connected through the Karush-Kuhn-Tucker (KKT) conditions which offer a way to check whether a solution is optimal for both the primal and dual problems.
%We consider the dual because: 1st we get rid of explicit dependence on $\mathcal{D}$, 2nd Fenchel conjugate gradient is smooth which is needed for Frank-wolfe
%The dual of (\ref{primal_alpha}) is the concave maximization problem 
\begin{equation}
\begin{aligned}
    &\max_{q \in \mbox{conv}({\cal Q})} \, \psi_{\alpha}(q),
    \end{aligned}
    \tag{$\textcolor{blue}{D_\alpha}$} 
    \label{dual_alpha}
\end{equation}
where  $\psi_\alpha(q):= \langle q, \mathcal{P} \rangle -\alpha  H^\ast \left( \frac{q}{\alpha}\right) $ and $H^\ast(y)=\log \left( \sum_{j} e^{y_j}\right)$ (the Fenchel conjugate of $H$).  As a consequence of Zalinescu's theorem \cite{zlinescu1983uniformly}, the Fenchel conjugate of $H$   is 1-smooth with respect to the $\|\cdot\|_\infty$, making \eqref{dual_alpha} a smooth concave optimization problem over a polyhedron.

Accordingly,  \eqref{dual_alpha}  can be expressed as 
maximization of
an  objective function  of the form $\mathbb{E}_{\mathbf{z} \sim \mathcal{P}}[\ell(q, \mathbf{z})]$, where $\ell(q,\mathbf{z}):=q(\mathbf{z})-\alpha H^\ast \left( \frac{q}{\alpha} \right)$ is both $L_0$-Lipschitz and $L_1$-smooth for suitable constants $L_0$ and $L_1$. Since the sets $\Delta_{k}$ and $\text{conv}(\mathcal{Q})$ are convex, the design of an algorithm solving the dual falls within the framework of Differentially Private Stochastic Convex Optimization (DP-SCO).
Lastly, according to Slater's theorem strong duality holds, i.e. $\phi^\ast= \psi^\ast$, where $\phi^\ast$ and $\psi^\ast$ denote minimum and maximum for \eqref{primal_alpha} and \eqref{dual_alpha}, respectively. 

In the following, we will focus on the regularized dual problem. We will later show how solving this problem leads to a differentially private solution of the unregularized primal, which is our main interest.

% ----------------- Synthetic data - end ------------------------------

% ----------------- DP Frank Wolfe - start ------------------------------

\section{Solving the dual: The differentially private Frank-Wolfe algorithm}
\label{sec:sol_dual}

This section establishes the differentially private Frank-Wolfe algorithm (DPFW) as a solution to the dual optimization problem \eqref{dual_alpha}.
An analysis of the corresponding solution to the unregularized optimization problem 
\eqref{primal}
is done in three steps, each corresponding to one of the following three subsections.
Section \ref{sec:FW_1} quantifies the difference in the values of the dual $\psi_{\alpha}$ evaluated in the true maximizer and its approximation through the output of DPFW.
Section \ref{sec:FW_2}  quantifies
the difference in the values  of the  primal function 
$\phi_{\alpha}$
evaluated in the true minimizer and its approximation resulting from
the output of  DPFW.
Section \ref{sec:FW_3}
considers
an optimal choice of the regularization parameter $\alpha$
and quantifies the overall accuracy of the resulting  value of the unregularized optimization problem \eqref{primal}.
Most notably, the quantification of accuracy in Section \ref{sec:FW_3}
is based on  the true distribution $\mathcal{P}$, not on the empirical distribution $\mathcal{P}_n$ associated with $S_n$.
For a high-level illustration of the individual steps leading to a differentially private solution of  \eqref{primal} see Figure \ref{fig:procedure}.

\subsection{Bounding the  dual gap}
\label{sec:FW_1}

The (non-privatized) Frank-Wolfe algorithm (FW) is an optimization algorithm designed for solving constrained convex optimization problems.
It aims at  minimization of  a convex (or maximization of a concave) and differentiable objective function over 
a compact and convex feasible set.
Since $\psi_{\alpha}$ is concave and $1$-smooth  and 
$\mbox{conv}(\mathcal{Q})$ compact and convex,
it is therefore particularly suited for solving  the dual \eqref{dual_alpha}. 
While competing methods such as (stochastic) gradient descent require a projection step back to the feasible set in each iteration,  FW only needs the solution of a convex problem over the same set in each iteration, and automatically stays in the feasible set. Under the constraints of the dual 
this is known to be computationally efficient
(\cite{jaggi2013frankwolfe},\cite{frank1956quadratic}).

The main idea of  FW is to linearize the objective function, resulting in the consideration of the following,  simpler optimization problem over the feasible set:
\begin{align}\label{linearization}
    \argmax\limits_{q \in {\cal Q}}  \langle \nabla\psi_\alpha(q_0), q \rangle
\end{align}
with initialization in some $q_0 \in \mbox{conv}(\mathcal{Q})$.

In the particular setting that we consider,
the gradient
\[
    \nabla \psi_\alpha (q)= \mathcal{P} - \nabla H^\ast \left( \frac{q}{\alpha} \right)
\]
 depends on the unknown data  distribution \(\mathcal{P}\), while we only have access to the dataset \(S = \{z_1, \ldots, z_n\} \stackrel{i.i.d.}{\sim} \mathcal{P}\). For an implementation of the algorithm,  \(\nabla \psi_\alpha\) is therefore approximated by
\begin{equation}
     \hat{\nabla}\psi_\alpha(q) := \mathcal{P}_n - \nabla H^\ast \left( \frac{q}{\alpha} \right),
     \label{eq:estimator_gradient}
\end{equation}
Moreover, 
for privatization of the algorithm, solving the linearized optimization problem  is randomized  by replacing \eqref{linearization} with the optimization problem 
\begin{align}
 s_0= \argmax\limits_{q \in {\cal Q}} \{  \langle \hat{\nabla}\psi_\alpha(q_0), q\rangle + u_{s, 0} \}, 
 \label{eq:laplace}
 \end{align}
where  $u_{s,0} \overset{i.i.d}{\sim} \text{Lap}(\lambda)$ and $\text{Lap}(\lambda)$ denotes a Laplace distribution with parameter $\lambda$; see Claim 3.9 in \cite{dwork2014algorithmic} (Report Noisy Max). 
An approximation of the maximizer $q_1$  is given by moving a step of size $\gamma$ in direction of the maximizer, i.e. 
\begin{align*}
    q_1=q_0+\gamma (s_0-q_0).
\end{align*}
Iteration of these steps results in a sequence of locations $q_0, q_1, \ldots$ approximating the maximizer of $\psi_{\alpha}$.
FW iterates until a stopping criterion is met (typically when the gradient's magnitude becomes small (indicating convergence), or after a fixed number of iterations).
In contrast to this, the  differentially private Frank-Wolfe algorithm (DPFW)  chooses its output $q^{out}_\alpha$ uniformly from 
$q_0, \ldots, q_{T-1}$, where $T$ corresponds to 
a predefined number of iterations.

The most significant difference between FW and DPFW, however, corresponds to the introduction of Laplace noise to the optimization step \eqref{eq:laplace}.
The intuition behind this so-called
\emph{Report Noisy Max} is to calibrate the amount of noise, i.e. 
the choice of the parameter $\lambda$,
 to the algorithm's sensitivity (i.e., how much its output changes in response to a single data point) thereby introducing privacy; for a more detailed description of \emph{Report Noisy Max} see Section 
 \ref{sec:report_noisy_max} in the appendix.

An algorithmic description of DPFW specifying   number of iterations, step size, and Laplace noise
as needed for the subsequently derived theoretical guarantees is provided by Algorithm \ref{alg:DPFW}.
The specifications of parameters in the algorithm depend on  
the diameter of the set \(\mathcal{Q}\) which, with respect to the \(\|\cdot\|_r\)-norm,   is denoted
 by \(D_r\).
 
\begin{algorithm}
\caption{Differentially private Frank-Wolfe algorithm (DPFW)
}
\begin{algorithmic}[1]
 %\SetAlgoLined
  \State \textbf{Input}: ($S_n$,  ${\cal Q}$,  $(\varepsilon, \delta)$,  $\alpha$)
  \State Number of iterations: $ T = \frac{D_1^{3/2} \varepsilon n}{\sqrt{32\alpha \log(1/\delta) }\log (2 |\mathcal{Q}|)}$\;
  \State Step size: $\gamma=2 \sqrt{\frac{\alpha}{TD_\infty}  }$\;
  \State Noise $\lambda = \frac{4 D_1\sqrt{2T \log(1/\delta)}}{\varepsilon n}$\;
\State Initialization: $q_0 \in \cal Q$\;
  \For{$t=0, \ldots, T-1$}
  \State Compute $\hat{\nabla}\psi_\alpha(q_t)$ according to \eqref{eq:estimator_gradient}
  %(or the exact gradient $\nabla_t=\nabla \psi_\alpha(q_t)$ for empirical) \;
  \State For all $s\in \mathcal{Q}$ sample $u_{s,t} \overset{i.i.d}{\sim}$ Lap ($\lambda$) 
  \State $s_t= \argmax\limits_{s \in {\cal Q}} \{  \langle \hat{\nabla}\psi_\alpha(q_t), s \rangle + u_{s,t} \} $ 
   \State  $q_{t+1}=q_t+ \gamma (s_t-q_t)$
  \EndFor
  \State \textbf{Output}: {$q^{\text{out}}_\alpha= q_U$ for $U \sim $ Uni$(\{0,\ldots,T-1\})$}
\end{algorithmic}
\label{alg:DPFW}
\end{algorithm}

The evaluation of the approximation of the maximizer of the dual \eqref{dual_alpha}
through the output of DPFW
is based on the following observation:
since $\psi_\alpha$ is concave and differentiable, it holds that
\begin{equation*}
\begin{aligned}
        0&\leq  \psi_\alpha(q^\ast_\alpha) - \psi_\alpha(q^{out}_\alpha)\leq \max_{q \in \mbox{conv}({\cal Q})}\langle \nabla \psi_\alpha(q), q^{out}_\alpha-q\rangle \;:=g_\mathcal{\psi_\alpha}(q^{out}_\alpha),
\end{aligned}
\end{equation*}
 where  $q^\ast_\alpha$ denotes the maximizer of $\psi_\alpha$.
We refer to $g_\mathcal{\psi_\alpha}(q^{out}_\alpha)$ as \textit{Frank-Wolfe Gap}. Naturally, an upper bound on $g_{\psi_\alpha}(q^{out}_\alpha)$ is also an upper bound for the  difference $\psi_\alpha(q^\ast_\alpha)-\psi_\alpha(q^{out}_\alpha)$ which quantifies the error
arising from utilization of DPFW to the end of maximizing $\psi_{\alpha}$. 
%Furthermore, since $g_{\psi_\alpha}(\cdot)$ is defined as the maximum of a linear function in $s$, and $\cal Q $ is finite, the maximum is reached at one of the vertexes, i.e at one of the points in $\mathcal{Q}$. 
The following theorem verifies $(\varepsilon, \delta)$-differential privacy of DPFW and gives an explicit upper bound for the Frank-Wolfe gap.

\begin{teo}
 The differentially private Frank -Wolfe algorithm DPFW  is  $(\varepsilon,\delta)$-DP. Moreover, the Frank-Wolfe Gap $g_\mathcal{\psi_\alpha}(q^{out}_\alpha):=\max_q\langle \nabla \psi_\alpha(q^{out}_\alpha), q^{out}_\alpha -q\rangle$  satisfies
   \begin{align}
     &\EE[g_ {\psi_\alpha}(q^\text{out}_\alpha)] =\mathcal{O}\left( \frac{ \log^{1/4} (1/\delta) \log^{1/2}(|\mathcal{Q}|)}{\alpha^{1/4} (\varepsilon n)^{1/2}} + D_1\sqrt{\frac{ \log(k)}{n}}\right) \;.  \notag 
   \end{align}
 \label{thm:utility_guarantee}
\end{teo}

Theorem \ref{thm:utility_guarantee} highlights the typical trade-off offered by DP algorithms: stronger privacy protection (smaller values of $(\varepsilon, \delta)$) results in weaker accuracy, while higher accuracy necessitates the reduction of noise, which in turn degrades privacy guarantees. Additionally, smaller values of  $\alpha$ indicate a larger Frank-Wolfe gap. At the same time, as stated in Section \ref{sec:regularization}, a solution to the regularized problem \eqref{primal_alpha}
is expected to be closer to a solution to the original optimization problem \eqref{primal} for small values of $\alpha$.
Moreover, if the number of queries is exponential, i.e $|\mathcal{Q}|=\exp(r)$ for some $r\geq 0$, the logarithmic factor $\log(|\mathcal{Q}|)$, enables answering them all in linear time.

Lastly, note that the second summand in \eqref{thm:utility_guarantee} arises from the evaluation  of the Frank-Wolfe gap against the  population loss. Although not explicitly shown in this work, our method, when evaluated in terms of empirical loss, provides the upper bound in Theorem \ref{thm:utility_guarantee} without the summand $D_1\sqrt{\log(k)/n}$.

\subsection{Bounding the primal gap}
\label{sec:FW_2}
In the previous section, we derived an approximate $(\varepsilon, \delta)$-DP solution to \eqref{dual_alpha} through the differentially private Frank-Wolfe algorithm DPFW. This section quantifies the deviation of the
 corresponding 
 value of the objective function $\phi$ in the primal  \eqref{primal}
  from the  actual optimum  
  $\phi^{\ast}$ of $\phi$.  
For this, we replace the maximizer in \eqref{primal_alpha} by the output 
 $q^{\text{out}}_\alpha$ of DPFW
and, subsequently,   analytically solve 
the resulting  optimization problem through implementation of the Karush–Kuhn–Tucker conditions. The  distribution solving the corresponding optimization problem is denoted as $\mathcal{P}^{\text{priv}}_\alpha$, i.e.
\begin{align}    \label{eq:whole_mechanism}
    \mathcal{P}_\alpha^{\text{priv}}=&
    \argmin\limits_{ \mathcal{D} \in \Delta_{\mathcal{Z}}}\mathcal{L}_{\alpha}(q_\alpha^{\text{out}}, \mathcal{D})=\argmin\limits_{ \mathcal{D} \in \Delta_{\mathcal{Z}}} \left( \langle q_\alpha^{\text{out}}, \mathcal{P} - \mathcal{D}\rangle+\alpha H(\mathcal{D}) \right)=:\mathcal{A}_{\alpha}(S_n).
\end{align}
It can be shown that the analytic solution to \eqref{eq:whole_mechanism} is $\mathcal{P}^{\text{priv}}_\alpha=\nabla H^\ast \left(\frac{q_\alpha^{\text{out}}}{\alpha} \right)$ (see Lemma \ref{KKT lemma} in the appendix).  To measure how much $\mathcal{P}^{\text{priv}}_\alpha$ deviates from the original $\mathcal{P}$, we estimate an upper bound on the primal gap, defined as follows: 
$$\textbf{Gap}_{\eqref{primal}} (\mathcal{D}) = \phi(\mathcal{D}) - \phi^\star \;,$$
where $\phi^{\star} := \min_{\mathcal{D} \in \Delta_{k}}\phi(\mathcal{D})$.  

\begin{figure}[ht]
    \centering
    \begin{tikzpicture}[node distance={4 cm}, thick,main/.style = {draw, circle}, scale=1.2, every node/.style={transform shape}]
        \node[] (1) {\eqref{primal}}; 
        \node[] (2) [right of=1] {\eqref{primal_alpha}};
        \node[] (3) [right of=2] {\eqref{dual_alpha}}; 
        \node[] (4) [above left of=1, node distance={2 cm}] {primal};
        \node [] (5) [below left of = 1, node distance={2cm}] {$\mathcal{P}^{\text{priv}}$};

        \draw[->] (1) to [out=60,in=120,looseness=0.7] node[midway, above] {regularized primal} (2);
        \draw[->] (2) to [out=60,in=120,looseness=0.7] node[midway, above] {regularized dual} (3);
        \draw[->] (3) to [out=210,in=270+60,looseness=0.7] node[midway, below] {$q^{\text{out}}_\alpha$ by DPFW } (2);
        \draw[->] (2) to [out=210,in=270+60,looseness=0.7] node[midway, below] {$\mathcal{P}^{\text{priv}}_\alpha$ by $\text{KKT}$ } (1);
        \draw[->] (4) to [out=270+60, in=120] (1);
        \draw[->] (1) to [out=270-30, in=30] (5);
    \end{tikzpicture}
    \caption{High-level illustration of the procedure to find $\mathcal{P}^{\text{priv}}$.}
    \label{fig:procedure}
\end{figure}
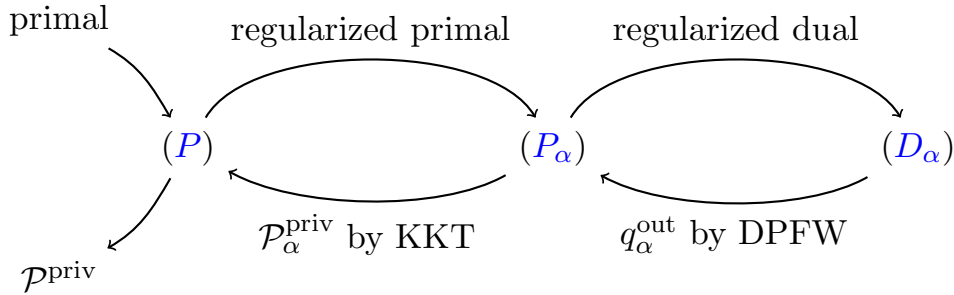

The primal  gap quantifies the deviation of a feasible point from its optimum. In the previous section, we established an upper bound on the Frank-Wolfe gap for the dual problem. This upper bound can be translated back to the primal setting, which is the content of Theorem \ref{thm:upper_bound_primalgap}. The entire process is summarized in Figure \ref{fig:procedure}. 

\begin{teo} For any $\alpha > 0$, we define $q^{\text{out}}_\alpha$ as the output of Algorithm \ref{alg:DPFW}, 
%which is executed on the dataset $S_n$ with the input parameters $\mathcal{Q}$, $\varepsilon$, $\delta$, and $\alpha$. 
Then, for the distribution $\mathcal{P}_\alpha^{\text{priv}} = \mathcal{A}_\alpha(S_n)$, obtained via \eqref{eq:whole_mechanism}, the following holds:
 \begin{equation}
  \EE[\textbf{Gap}_{\eqref{primal}}(\mathcal{P}^{\text{priv}}) ] \leq  \EE[g_{\psi_\alpha}(q^{\text{out}}_\alpha)] +\alpha \log k \;.  
    \label{gap1}
\end{equation}  
\label{thm:upper_bound_primalgap}
\end{teo}

\subsection{Bounding the overall gap}
\label{sec:FW_3}
Theorems  \ref{thm:utility_guarantee} and \ref{thm:upper_bound_primalgap}  establish  upper bounds on 
dual and primal gap of the proposed privacy-preserving algorithm solving the optimization problem \eqref{primal}.
By choosing $\alpha^\ast$ as the value that minimizes the right-hand side of \eqref{gap1}, we obtain 
\begin{equation}
    \label{eq:optimal_alpha}
    \alpha^\ast = \frac{ \log^{1/5}(1/\delta) \log^{2/5}|\mathcal{Q}|}{(\varepsilon n)^{2/5} \log^{4/5} k}\;
\end{equation}
and, accordingly, we define
$\mathcal{P}^{\text{priv}} = \mathcal{P}^{\text{priv}}_{\alpha^\ast}$.
The following theorem summarizes these findings:
\begin{teo} 
\label{MAINRESULTS}
 %Let $S_n = \{z_1, \ldots, z_n\}$ be a dataset of $n$ elements sampled independently from an unknown distribution $\mathcal{P}$, which is supported on  $\mathcal{Z}$. Additionally, let $\mathcal{Q}$ be a finite and symmetric set of queries. 
 
 Let $q^{\text{out}}_{\alpha^\ast}$ be the output of Algorithm \ref{alg:DPFW}, and  let $\alpha^\ast$ be the value defined in \eqref{eq:optimal_alpha}. Then, the distribution $\mathcal{P}^{\text{priv}} = \mathcal{A}_{\alpha^\ast}(S_n)$, obtained through \eqref{eq:whole_mechanism}, guarantees that the maximum error in answering all queries from $\mathcal{Q}$ using $\mathcal{P}^{\text{priv}}$ is bounded as follows:
\begin{align*}
&  \EE\left[\max\limits_{q\in \mathcal{Q}}\langle q,\mathcal{P}- \mathcal{P}^\text{priv}\rangle \right] = \mathcal{O}\left( \frac{ \log^{1/5}(1/\delta) \log^{2/5} |{\cal Q}|\log^{1/5} k}{(\varepsilon n)^{2/5} }  +D_1\sqrt{\frac{  \log k}{n}}\right) \;.
\end{align*}
 \end{teo}
%In many practical applications,  $\mathcal{Z} = \{0,1\}^d $ (see also Section \ref{sec:preliminaries}).  Imagine, for example,  each element $z\in S_n$ corresponds to an individual characterized  by $d$ features. In this case, a value of 0 in the $i$-th postion of $z$ indicates the absence of a feature, while a value of 1 indicates its presence.
We consider the specific setting  \(\mathcal{Z} = \{0,1\}^d\) (see also Section \ref{sec:preliminaries}) to derive an upper bound for comparison with other algorithms (see Table \ref{tab:synthetic_algorithms} in Section \ref{sec:conclusion}).
For this, we do not consider the population risk $\max_{q\in \mathcal{Q}}\langle q,\mathcal{P}- \mathcal{P}^\text{priv}\rangle$, but its empirical analogue $\max_{q\in \mathcal{Q}}\langle q,\mathcal{P}_n- \mathcal{P}^\text{priv}\rangle$.
In this case, based on Theorem \ref{MAINRESULTS},  the following bound can be derived:
\[
\EE[\textbf{Gap}_{\eqref{primal}}(\mathcal{P}^{\text{priv}})] = \mathcal{O} \left( \frac{\log^{2/5}(|{\cal Q}|)\log^{1/5}(1/\delta) d^{1/5}}{(\varepsilon n)^{2/5}} \right)\;.
\]
Note that this term does not include the summand \(D_1\sqrt{\log k / n}\), arising from using the unbiased estimator \(\hat{\nabla} \psi_\alpha\) for \(\nabla \psi_\alpha\). In fact, replacing \(\mathcal{P}\) with the empirical distribution \(\mathcal{P}_n\) from \(S_n\), allows exact computation of \(\nabla \psi_\alpha(q) = \mathcal{P}_n - \nabla H^\ast(q/\alpha)\), eliminating the need for approximation.

\section{Solving the primal: Randomized Smoothing and Stochastic Composite Minimization}
\label{sec:primal}

This section provides yet another solution to the unregularized optimization problem \eqref{primal}.
In contrast to Section \ref{sec:sol_dual}, which focused on a solution by solving the regularized dual \eqref{dual_alpha}, this section provides a solution by solving the regularized primal \eqref{primal_alpha}.
For this, note that  \eqref{primal_alpha} can be written as
\begin{equation}
\begin{aligned}
    \min\limits_{\mathcal{D} \in \Delta_{k} }\left[\phi(\mathcal{D})+ \alpha H(\mathcal{D})\right] \quad 
\end{aligned}
\tag{$\textcolor{blue}{P_\alpha}$}
\end{equation}
where $\phi(\mathcal{D}):= \max_{q \in \text{conv}(\mathcal{Q})}\langle q, \mathcal{P}-\mathcal{D} \rangle$, and  $H(\mathcal{D})=\sum_{z\in \mathcal{Z}} \mathcal{D}(z) \log (\mathcal{D}(z))$ is the negative entropy of $\mathcal{D} \in\Delta_{{k}}$.
If $\phi$ and $H$ in  
    \eqref{primal_alpha}
satisfied  specific smoothness and convexity properties, the optimization problem would meet the requirements summarized as stochastic composite oracle
model in \cite{daspremont:hal-04230893} and could be solved by the so-called (non-)accelerated complementary composite
stochastic mirror-descent, an optimization technique proposed in  \cite{daspremont:hal-04230893}.
Non-smoothness of the inner maximization problem in \eqref{primal_alpha}, i.e.  non-smoothness of $\phi$, however, prevents from direct application of corresponding results.
To resolve this issue, we pursue  randomized smoothing 
of $\phi$ (Section \ref{eqn:rand_smooth}) and a subsequent application of  Algorithm 3.2 in \cite{daspremont:hal-04230893},
guaranteeing fast optimization rates (Section \ref{sec:stoch_comp_min}).

\subsection{Randomized Smoothing}\label{sec:rand_smooth}

We resolve non-smoothness of $\phi$ by
a convolution-based smoothing technique amenable to non-smooth stochastic optimization problems; see, for example, \cite{DBW}.  The intuition underlying this approach is that convolving two functions results in a new function that is at least as smooth as the smoothest of the two original functions. 
For our purposes, 
we consider convolution of $\phi$
with  the density of a multivariate 
normal distribution with mean vector $\mathbf{0}$ and covariance matrix $\sigma^2I_{k}$, $\sigma>0$, 
denoted by $\varphi_{\sigma}$.
Accordingly, we consider the smoothed objective function $\phi_{\sigma}$ defined as follows: 
\begin{equation}\label{eqn:rand_smooth} \phi_{\sigma}(\mathcal{D}):=\int_{\mathbb{R}^{k}}\phi(\mathcal{D}+x)\varphi_{\sigma}(x)\,dx=\mathbb{E}_{\xi} [\phi(\mathcal{D}+\xi)]
\end{equation}
where $\xi$ is a random vector in $\RR^{k}$ with probability density $\varphi_{\sigma}$.
In general, $\phi_{\sigma}$ is convex and differentiable whenever $\phi$ is convex.

Convexity of $\phi$ is established by the following lemma:
\begin{lema}
$\phi:{\Delta}_{k}\mapsto \mathbb{R}_+$ is convex and 1-Lipschitz with respect to~$\|\cdot\|_1$.
\label{lema:phi_Lipschitz}
\end{lema}

 Properties of $\phi_\sigma$ inherited from $\phi$ are summarized by the following proposition.
 
\begin{prop}
Let $\mathcal{Q}\subseteq[-1,1]^k$ be a compact and convex set, and  let $\phi(\mathcal{D}):= \max\limits_{ q\in \mathcal{Q}} \langle q, \mathcal{P}-\mathcal{D}\rangle$, the function $\phi_{\sigma}$ defined in ~\eqref{eqn:rand_smooth}, satisfies:
\begin{enumerate}
    \item[(i)] $\|\phi_{\sigma}-\phi\|_{\infty}\leq \sigma w(\mathcal{Q})$, where $w(\mathcal{Q}):=\mathbb{E}_{\xi\sim \varphi}\left[\max\limits_{q\in \mathcal{Q}}\langle q,\xi\rangle\right]$ is the Gaussian width of $\mathcal{Q}$ and $\varphi$ denotes the density of the  multivariate standard normal distribution.
normal distribution 
    \item[(ii)] $\phi_{\sigma}$ is convex and 1-Lipschitz with respect to~$\|\cdot\|_1$.
    \item[(iii)] $\phi_{\sigma}$ is $1/\sigma$-smooth with respect to $\|\cdot\|_1$.
\end{enumerate}
\end{prop}

\subsection{Stochastic Composite Minimization}
\label{sec:stoch_comp_min}

Randomized smoothing and regularization allow for  reformulation of the primal problem \eqref{primal_alpha} into
 the following  approximation within the complementary composite framework considered in \cite{daspremont:hal-04230893}:
\begin{equation} \label{eqn:primal_rand_smooth} 
\min_{\mathcal{D }\in \Delta_{k}} \Phi_{\alpha,\sigma}(\mathcal{D}): = \min_{\mathcal{D }\in \Delta_{k}} \left[\phi_{\sigma}(\mathcal{D}) + \alpha H(\mathcal{D}) \right],
\tag{$\textcolor{blue}{P_{\alpha,\sigma}}$}
\end{equation}
where $\phi_\sigma$ is $\frac{1}{\sigma}$-smooth and $H$ is $1$-strongly convex.

Given a Gaussian vector $\xi \sim \mathcal{N}(0, \sigma^2 I_k)$, a stochastic first-order oracle for $\phi_{\sigma}$ can be constructed as
\begin{equation}
\label{stochastic-oracle}
G_{\sigma} (\mathcal{D},\xi, S_n, \mathcal{Q}) \in \argmax_{q \in \text{conv}(\mathcal{Q})} \langle q , \mathcal{P}_n - \mathcal{D} + \xi \rangle.
\end{equation}
As shown in Section \ref{appendix-stochastic-oracle} of the appendix, this oracle is an unbiased estimator of $\nabla \phi_\sigma$ and has bounded variance. We notice that we cannot use the usual gradient descent on $\Phi_{\alpha,\sigma}$ for two reasons, the first is that our problem is constrained in the simplex $\Delta_k$ and the second is that we only have gradients related to $\nabla \phi_\sigma$ and not to the whole  $\nabla \Phi_{\alpha,\sigma}$. 

Solving equation \eqref{eqn:primal_rand_smooth} requires tools from the composite setting. For that, we claim that for any coefficients $A,B > 0,$ and vectors $g \in \mathbb{R}^{k}, \mathcal{D}'\in \Delta_k $, the following problem has a close form solution
\begin{align*}
\min_{\mathcal{D} \in {\Delta_k}} \Big\{A[\langle g  , \mathcal{D} \rangle + H(\mathcal{D})] + B D_{H}(\mathcal{D},\mathcal{D}') \Big\} 
\end{align*}
with $H$ being the negative entropy function and $D_{H}$ being the Kullback-Leibler divergence (Proof in Section \ref{proof-close-form} of the appendix). Now we have verified all the assumptions of Theorem 3.4 in \cite{daspremont:hal-04230893}  and we provide an $(\varepsilon, \delta)$-DP variation of the Accelerated Composite Mirror Descent (Algorithm \ref{Algo ACSMD}).

\begin{algorithm}
\caption{Differentially Private Complementary Accelerated Mirror-Descent }
\label{Algo ACSMD}
\begin{algorithmic}[1] % [1] adds line numbers
\State \textbf{Input}: ($S_n$,  ${\cal Q}$,  $(\varepsilon, \delta)$,  $\alpha$)
\State Number of iterations: $T  =\frac{\log^{1/2} k}{\log^{1/2}(1/\delta)} \frac{\varepsilon n}{w(\mathcal{Q})}$
\State Smoothing parameter: $\sigma  =\frac{4\sqrt{T\ln(1/\delta)}}{n\varepsilon}$ 
\State Step sizes: for $1 \leq t \leq T, \quad \eta_t = t + \sqrt{\frac{4}{\alpha \sigma }} +1 $
\State Initialisation $\mathcal{D}_1 = \mathcal{D}_1^{ag} = [1/k,\ldots,1/k]^{\top}$
\For{$1 \leq t \leq T$}
    \State $\mathcal{D}_{t}^{md}  = \frac{\sum_{\tau = 1}^{t-1} \eta_\tau}{\sum_{\tau = 1}^{t} \eta_\tau} \mathcal{D}_{t}^{ag} + \frac{ \eta_t}{\sum_{\tau = 1}^{t} \eta_\tau} \mathcal{D}_{t}$
    \State $g_t  = G_{\sigma} (\mathcal{D}_{t}^{md}, \xi_t, S_n, \mathcal{Q})\,, \xi_t \sim{\cal N}(0,\sigma^2I_k)$ 
    \State $\mathcal{D}_{t+1}  = \argmin_{\mathcal{D} \in {\Delta_k}} \Big\{\eta_t [\langle g_t  , \mathcal{D} \rangle + H(\mathcal{D})] + \left(\sum_{\tau = 1}^{t-1} \eta_\tau \right) D_{H}(\mathcal{D},\mathcal{D}_{t}) \Big\}$
    \State $\mathcal{D}_{t+1}^{ag}  = \frac{\sum_{\tau = 1}^{t-1} \eta_\tau}{\sum_{\tau = 1}^{t} \eta_\tau}  \mathcal{D}_{t}^{ag} + \frac{ \eta_t}{\sum_{\tau = 1}^{t} \eta_\tau} \mathcal{D}_{t+1}$
\EndFor
\State \textbf{Output}: $\mathcal{P}^{\text{priv}}_\alpha = \mathcal{D} _{T+1}^{ag}$. 
\end{algorithmic}
\end{algorithm}
 
Algorithm \ref{Algo ACSMD} combines mirror descent with accelerated gradient descent. Mirror descent performs gradient updates in the dual space, which is particularly useful in non-Euclidean domains where standard updates may lead to infeasible or geometrically meaningless solutions (e.g., on manifolds).

Our algorithm works for any \(\alpha \geq 0\), but we offer a simplified interpretation for the case \(\alpha = 0\). Starting from an initial \(\mathcal{D}_1\), at the \(t\)-th iteration, the current distribution \(\mathcal{D}_t\) is mapped to the dual space via the mirror map \(\nabla H\). There, the update is \(\nabla H(\mathcal{D}_t) - \frac{\eta_t}{\sum_{\tau=1}^{t-1} \eta_\tau} g_t\), where \(g_t = G_{\sigma}(\mathcal{D}_t, \xi_t, S_n, \mathcal{Q})\) approximates \(\nabla \Phi_{\alpha, \sigma}\). The updated value is then back to the primal space using the inverse mirror map \(\nabla H^{-1} = \nabla H^{\ast}\). Thus, the update in the primal space is:
\[
\mathcal{D}_{t+1} = \nabla H^{\ast} \left( \nabla H(\mathcal{D}_t) - \frac{\eta_t}{\sum_{\tau=1}^{t-1} \eta_\tau} g_t \right)
\]
As demonstrated in Theorem 6.13 of \cite{intro-orobona}, this update is equivalent to
\begin{align*}
\mathcal{D}_{t+1} = \argmin\limits_{\D \in \Delta_k } \left\{  \eta_t \langle g_t, \D\rangle  +  \left(\sum_{\tau=1}^{t-1} \eta_{\tau}\right)D_{H} (\D, \D_t)\right\}\;.
\end{align*}

In general, accelerated methods improve upon the convergence rate of standard first-order methods for convex optimization problems by exploiting the convexity and the smoothness of the objective function. For this, it  combines the traditional mirror descent updates with an additional sequence of gradient evaluations at ($D_t^{md}$) to add a momentum to the mirror descent. While regular mirror descent takes steps based solely on the current gradient, accelerated mirror descent adds information from previous iterations (aggregation), allowing for faster convergence. 

Differential Privacy is introduced through the noisy oracle $ G_\sigma $, which injects calibrated noise at each iteration, ensuring that the Algorithm is $(\varepsilon, \delta)-$DP. The next result provides the accuracy and the privacy guarantee for Algorithm \ref{Algo ACSMD}.
\begin{teo}
\label{thm:RS-upperbond1}
 Algorithm \ref{Algo ACSMD} is $(\varepsilon,\delta)$-DP and the output $\mathcal{P}^{\text{priv}}_\alpha$  satisfies
\begin{align*}
&\EE\left[\textbf{Gap}_{\eqref{eqn:primal_rand_smooth}}(\mathcal{P}^{\text{priv}}_\alpha)  \right]  =  \mathcal{O} \left(
 \frac{\log^{3/4}(1/\delta)}{\alpha \log^{5/4}(k)}  \frac{w^{5/2}(\mathcal{Q})}{(n\varepsilon)^{3/2} }+ \frac{\log^{1/2}(1/\delta)}{\alpha \log^{1/2}(k)} \frac{w(\mathcal{Q})}{n \varepsilon}  \right).
\end{align*}
\end{teo}
Similarly to the approach taken with DPFW, optimizing over $\alpha$ yields
\begin{teo}
\label{thm:smoothing}
Let us consider the accuracy guarantee obtained in Theorem \ref{thm:RS-upperbond1}.

For \( n \geq \frac{w^3(\mathcal{Q}) \log(1/\delta)}{\epsilon \log^{3/2}(k)}\)
and \(\alpha^\ast = \frac{\log^{1/2}(1/\delta) w^{1/2}(\mathcal{Q}) }{\log ^{3/4}(k) \sqrt{n \varepsilon}} \), the distribution  $\mathcal{P}^{\text{priv}}:=\mathcal{P}_{\alpha^\ast}^{\text{priv}} $ satisfies
 \begin{align*}
 &\EE\left[\max\limits_{q\in \mathcal{Q}} \, \langle q, {\mathcal{P}} - \mathcal{P}^{\text{priv}} \rangle\right] =\mathcal{O}
\left( \frac{w^{1/2}(\mathcal{Q}) \log^{1/4}(k)\log^{1/4}( 1/\delta)}{ \varepsilon^{1/2} n^{1/2}} \right)\;.
 \end{align*}
\end{teo}
In the case the query class is contained in the  Euclidean unit ball $\mathcal{B}_2^k:=\left\{ x \in \mathbb{R}^k : \|x\|_2 \leq 1 \right\}$, the previous upper bound provides the strongest accuracy guarantees compared to the methods listed in Table \ref{tab:synthetic_algorithms}, particularly as $n$ increases.  More specifically, if   $\mathcal{Q}$ is a finite subset of $\mathcal{B}_2^k$, then $w(\mathcal{Q}) \leq C \log |\mathcal{Q}|\;,$ where $C$ is a constant independent of the dimension $k$. 

\section{Conclusion}
\label{sec:conclusion}

This work presents the first significant population risk bounds for differentially private synthetic data. We introduce two algorithms: DP-Frank-Wolfe (DPFW) and DP-Randomized Smoothing with mirror descent (DPAM). Both algorithms achieve tight worst-case error bounds that are competitive with state-of-the-art methods, as shown in Table \ref{tab:synthetic_algorithms}. Notably, DPAM maintains the worst-case error for both empirical and population risk. 
For a specific class of queries, DPAM outperforms all previous methods in terms of sample size \(n\) (with a rate of \(1/\sqrt{n}\)), offering stronger guarantees by incorporating smaller factors related to other parameters compared to the established algorithms.

We conclude by pointing towards future research on establishing optimality of upper bounds. In the appendix we discuss the results of \cite{Bun:2018} on private lower bounds, and we conjecture that with proper control of the Gaussian width, DPAM is optimal up to a factor $\log(1/\delta)$.

\begin{table}[t]
\centering
\setlength{\tabcolsep}{4pt} % Adjust the column separation
\begin{tabular}{p{2cm}| p{5.8cm}} % Set fixed widths for the column
\toprule
\textbf{Established} & \textbf{Empirical Loss} \\
\midrule
MWEM & \( \mathcal{O}\left(\frac{d^{1/4} \log^{1/2} |\mathcal{Q}| \log^{1/2} (1/\delta)}{n^{1/2} \epsilon^{1/2}}\right) \) \\
DualQuery & \( \mathcal{O}\left(\frac{d^{1/6} \log^{1/2} |\mathcal{Q}| \log^{1/6} (1/\delta)}{n^{1/3} \epsilon^{1/3}}\right) \) \\
FEM & \( \mathcal{O}\left(\frac{d^{3/4} \log^{1/2} |Q| \log^{1/2} (1/\delta)}{n^{1/2} \epsilon^{1/2}}\right) \) \\
sepFEM & \( \mathcal{O}\left(\frac{d^{5/8} \log^{1/2} |\mathcal{Q}| \log^{1/2} (1/\delta)}{n^{1/2} \epsilon^{1/2}}\right) \) \\
DQRS & \( \mathcal{O}\left(\frac{d^{1/5} \log^{3/5} |\mathcal{Q}| \log^{1/5} (1/\delta)}{n^{2/5} \epsilon^{2/5}}\right) \) \\
DPFW(E) & \( \mathcal{O} \left(  \frac{d^{1/5}\log^{2/5}|{\cal Q}|\log^{1/5}(1/\delta) }{n^{2/5} \varepsilon^{2/5}}\right) \) \\
\midrule
\textbf{New} & \textbf{Population Loss} \\
\midrule
DPFW(P) & \( \mathcal{O} \left(  \frac{d^{1/5}\log^{2/5}|{\cal Q}|\log^{1/5}(1/\delta) }{n^{2/5} \varepsilon^{2/5}} +D_1\frac{\log^{1/2}(d)}{n^{1/2}}\right) \) \\
DPAM & \( \mathcal{O} \left( \frac{d^{1/4} w^{1/2}(\mathcal{Q}) \log^{1/4}(1/\delta)}{n^{1/2}\varepsilon^{1/2}} \right) \)\\
\bottomrule
\end{tabular}
\caption{
This table provides upper bounds on the empirical loss  \(\max_{q \in \mathcal{Q}} \langle q, \mathcal{P}_n - \mathcal{P}^{\text{priv}} \rangle\)  and the population loss \(\max_{q \in \mathcal{Q}} \langle q, \mathcal{P} - \mathcal{P}^{\text{priv}} \rangle\) of established on the two newly proposed algorithms.
For this, the data universe is  assumed to be \(\mathcal{Z} = \{0,1\}^d\). 
}  
\label{tab:synthetic_algorithms}
\end{table}

\section*{Acknowledgments}
The authors gratefully acknowledge Cristóbal Guzmán's support and, in particular, his invaluable  insights in the early stages of this work.
Annika Betken gratefully acknowledges financial support from the Dutch Research Council (NWO) through VENI grant 212.164. 

\nocite{*}

\bibliography{ref}

\newpage
\appendix
\onecolumn
\begin{center}
    \Large{\textbf{Appendix}}
\end{center}

This part provides auxiliary results needed for the proofs of the mathematical results in Section \ref{sec:prelim_A}. It, moreover, establishes the proofs of Theorems \ref{thm:utility_guarantee} and \ref{thm:upper_bound_primalgap} stated in Section \ref{sec:sol_dual} of Section \ref{sec:supp_A}, and the proof of Theorems  \ref{thm:RS-upperbond1}
and \ref{thm:smoothing}
stated in Section \ref{sec:primal}
 of the Section \ref{sec:proofs_5}.
The last section (Section \ref{sec:lower_bounds}) provides a short discussion of lower bounds for the accuracy of algorithms for private synthetic data release.

\paragraph{Notation:}

For  $\|\cdot\|$ a norm and $f:E \longrightarrow \RR$ convex (over $E\subset \RR^d$ convex), we denote with $f^\ast(y)=\sup_{x\in E} (\langle x,y\rangle - f(y))$ the Fenchel conjugate of $f$. The dual norm of  $\|\cdot\|$ is denoted by $\|x\|_\ast=\sup_{\|y\|\leq 1} \langle x, y \rangle\;.$ For $f$, we denote with $\nabla f$ and $\nabla^2 f$ the gradient vector and Hessian matrix of $f$, respectively.  For a vector $v=(v_1, \ldots , v_d)^\intercal$, the symbol diag$(v)$ denotes a diagonal matrix whose elements on the main diagonal are given by the elements of $v$ and all the remaining entries are set to zero. 
%When $X_1, \ldots, X_n \overset{i.i.d}{\sim} g$, we mean that the random variables $(X_i)_{i=1}^n$ are all independent and identically distribued, with probability density function $g$. 
For $\mathbf{z}\in \mathcal{Z}=\{z_1, \ldots, z_k\}$, we define the vector $e_\mathbf{z}:= \left( \mathbf{1}_{\{\mathbf{z}=z_1\}},\ldots, \mathbf{1}_{\{\mathbf{z}=z_k\}} \right)^\intercal\;.$  
Any notation not explicitly introduced in this paragraph adheres to the conventions established in the main document. 

\section{Auxiliary results}\label{sec:prelim_A}

Both algorithms established in this article are based on the consideration of an optimization problem of the form
\begin{align*}
\min_{\mathcal{D} \in {\Delta_k}} \Psi_{A,B,C}(\mathcal{D}), \ \ \Psi_{A,B,C}(\mathcal{D}): =  A \langle g  , \mathcal{D} \rangle + B H(\mathcal{D}) + C D_{H},(\mathcal{D},\mathcal{D}'),
\end{align*}
 $A\in \mathbb{R}$, $B > 0$, $C \geq 0 $,   $g \in \mathbb{R}^{k}$, $\mathcal{D}'\in \Delta_k$, and where 
 $H$ denotes  the negative entropy function on $\Delta_k$   and $D_H$ the associated Bregman divergence.
A sufficient condition  guaranteeing  a unique solution to an optimization problem of this type, corresponds to strong convexity of the objective  function (here: $\Psi_{A,B, C}$).
For this, note that the first summand in the representation of $\Psi_{A,B, C}$ is linear, while the second summand corresponds to the negative entropy function which is strongly convex according to Example 2.5 in \cite{shalev2012online}.
For establishing strong convexity of $\Psi_{A,B, C}$,  it therefore suffices to show that   $D_{H}(\cdot ,D')$
is (strongly) convex. In fact, it can be shown that any strongly convex function $f$
 induces strong convexity of the Bregman divergence in its first argument:
 
\begin{lema}\label{lem:strong_conv_div}
Let \( f: \mathbb{R}^n \to \mathbb{R} \) be a differentiable, \( \mu \)-strongly convex function, i.e. for all \( x, y \in \mathbb{R}^n \),
\[
f(x) \geq f(y) + \langle \nabla f(y), x - y \rangle + \frac{\mu}{2} \| x - y \|^2.
\]
Then, the Bregman divergence \( D_f(x, y) \) 
is strongly convex in \( x \), with convexity parameter \( \mu \). 
\end{lema}

\begin{proof}[Proof of Lemma \ref{lem:strong_conv_div}]
 Since \( f \) is \( \mu \)-strongly convex, for any \( x_1, x_2 \in \mathbb{R}^n \) and \( \lambda \in [0, 1] \), we have
   \[
   f(\lambda x_1 + (1 - \lambda) x_2) \leq \lambda f(x_1) + (1 - \lambda) f(x_2) - \frac{\mu}{2} \lambda(1 - \lambda) \| x_1 - x_2 \|^2.
   \]
Since the gradient is linear, for all \( y \in \mathbb{R}^n \),
   \[
   \langle \nabla f(y), \lambda x_1 + (1 - \lambda) x_2 - y \rangle = \lambda \langle \nabla f(y), x_1 - y \rangle + (1 - \lambda) \langle \nabla f(y), x_2 - y \rangle.
   \] 
   
Using the above properties, we calculate the Bregman divergence for the convex combination \( \lambda x_1 + (1 - \lambda) x_2 \) as follows:
   \[
   \begin{aligned}
   D_f(\lambda x_1 + (1 - \lambda) x_2, y) 
   &= f(\lambda x_1 + (1 - \lambda) x_2) - f(y) - \langle \nabla f(y), \lambda x_1 + (1 - \lambda) x_2 - y \rangle \\
   &\leq \left( \lambda f(x_1) + (1 - \lambda) f(x_2) - \frac{\mu}{2} \lambda(1 - \lambda) \| x_1 - x_2 \|^2 \right) \\
   &\quad - f(y) - \left( \lambda \langle \nabla f(y), x_1 - y \rangle + (1 - \lambda) \langle \nabla f(y), x_2 - y \rangle \right) \\
   &= \lambda D_f(x_1, y) + (1 - \lambda) D_f(x_2, y) - \frac{\mu}{2} \lambda(1 - \lambda) \| x_1 - x_2 \|^2.
   \end{aligned}
   \]
Thus, the Bregman divergence \( D_f(x, y) \) is strongly convex in its first argument \( x \) with parameter \( \mu \).
\end{proof}

Section \ref{sec:sol_dual} aims at solving the optimization problem 
\begin{align}
    &\max_{q \in \mbox{conv}({\cal Q})} \, \psi_{\alpha}(q),
    \end{align}
where  $\psi_\alpha(q):= \langle q, \mathcal{P} \rangle -\alpha  H^\ast \left( \frac{q}{\alpha}\right)$.
Paving the way for the proofs of the results reported in Section \ref{sec:sol_dual}, we state the following auxiliary lemma establishing properties of $\psi_\alpha$ needed for the proof of Theorem \ref{thm:utility_guarantee}.

\begin{lema}\label{smoothness_dual}
The function
 $\psi_\alpha: [-1,1]^k\longrightarrow \mathbb{R}$,  $\psi_\alpha$ is $2$-Lipschitz and  $\frac{1}{\alpha}$-smooth with respect to $\|\cdot\|_\infty$.
\end{lema}

\begin{proof}[Proof of Lemma \ref{smoothness_dual}] Note that  $\psi_\alpha (q):= \EE_{\mathbf{z}\sim \mathcal{P}}[\ell(q,\mathbf{z})]$ where $\ell(q,\mathbf{z}):=q(z)-\alpha H^\ast \left( \frac{q}{\alpha} \right)$. 
Let $\partial_j f $ denote the partial derivative with respect to the $j-$th entry of an $\RR^k$-valued function $f$. Then, it holds  that $\partial_j q(\mathbf{z})= \delta_{\mathbf{z},j}$, where $\delta_{a,b}=1$ if and only if $a=b$, and 0 else.  Note that $\partial_j q(\mathbf{z})$ is integrable, such that, by the dominated convergence theorem, it follows that 
\begin{align*}
  \nabla   \EE_{\mathbf{z}\sim \mathcal{P}}\left[\ell(q, \mathbf{z})\right]=\EE_{\mathbf{z}\sim \mathcal{P}}\left[ \nabla\ell(q, \mathbf{z})\right]
\end{align*}
Accordingly,  it suffices to show that $\ell(\cdot,\mathbf{z})$ is $2-$Lipschitz and $\frac{1}{\alpha}-$smooth.

\begin{enumerate}
    \item \textbf{Lipschitzness}:  For   $q_1, q_2\in [-1,1]^k$, it holds that
\begin{align*}
   | \ell(q_1, \mathbf{z})- \ell(q_2, \mathbf{z})  =& \left| q_1({\mathbf{z}}) -  q_2({\mathbf{z}}) + \alpha \left( H^\ast \left( \frac{q_2}{\alpha} \right) - H^\ast \left(\frac{q_1}{\alpha}\right) \right) \right|\\
   \leq& \sup_{\mathbf{z}\in \mathcal{Z}} \Big| q_1({\mathbf{z}}) -  q_2({\mathbf{z}}) \Big| +\alpha \left\| \frac{q_2}{\alpha}- \frac{q_1}{\alpha}\right\|_\infty \leq 2\| q_1- q_2\|_{\infty}\;.
\end{align*}
\item \textbf{Smoothness}: 
For   $q_1, q_2\in [-1,1]^k$, it holds that
\begin{align*}
\| \nabla \ell(q_1, \mathbf{z}) - \nabla \ell(q_2, \mathbf{z}) \|_1 = \Big\| e_\mathbf{z} - \nabla [H^\ast (q_1 /\alpha) ]-  e_\mathbf{z} + \nabla[ H^\ast (q_2 /\alpha)]\Big\|_1 \leq \frac{1}{\alpha} \|q_1-q_2\|_\infty \;.
\end{align*}
%The conclusion follows from linearity of the expectation and Jensen's inequality: for every norm $\|\cdot\|$ and random vector $X$,  $\| \EE[X]\|\leq \EE[ \|X\|]\;.$
\end{enumerate}
\end{proof}

Algorithm  \ref{alg:DPFW} incorporates Laplace noise 
for guaranteeing differential privacy. 
For an accuracy guarantee of the algorithm, however, we need to control the effect of this mechanism. For this, 
we establish the following
maximal inequality for i.i.d. Laplace random variables.

\begin{lema} Consider $u_1, \ldots, u_k\overset{i.i.d}{\sim} \text{Lap}(\lambda)$ random variables. Then, it holds that 
 \[ \EE\left[ \max_{j=1,\ldots, k} u_j \right] \leq  2\lambda \log \left( 2k \right)\;.\]
 \label{lemma:properties_laplacians}
 \end{lema}
\begin{proof}[Proof of Lemma \ref{lemma:properties_laplacians}]
For all $t>0$, 
$$ \exp{\left(t \EE \left[ \max_{j=1,\ldots, k} u_j \right]\right)} \leq \EE \left[ \exp{(t \max_{j =1,\ldots, k} u_j})\right] \leq \EE \left[ \sum_{i=1}^{k} \exp{(t u_i)}\right]= k \EE[ \exp{(t u_1)}].$$
The first inequality follows by Jensen's inequality applied to $\exp(\cdot)$, which is convex.  Applying  $\log (\cdot)$ to both sides ($\log$ is monotone and preserves the inequality)
$$\EE \left[ \max_{j =1,\ldots, k} u_j \right] \leq \frac{\log \left( k  \EE[ \exp{(t u_1)}]\right)}{t} \quad \text{for all }t>0.$$

Since the density function of the Laplacian distribution is known, $\EE[ \exp{(t u)}]$ can be computed exactly. In fact, it corresponds to the moment generating function. It is finite if $0< t< \frac{1}{\lambda}$, since 
$$\EE[ \exp{(t u)}]=\frac{1}{(1+t\lambda)(1-\lambda t)}\leq \frac{1}{1-\lambda t} \quad \text{ for all }0< t< \frac{1}{\lambda}. $$
By choosing $t=\frac{1}{2\lambda}$, we finally get the upper bound
\( \EE \left[ \max\limits_{j=1,\ldots, k} u_j \right] \leq  2\lambda \log \left( 2k \right)\;. \)
\end{proof}

Section \ref{sec:primal} aims at solving 
 the unregularized optimization problem 
 \begin{equation*}
\min\limits_{\mathcal{D} \in \Delta_{k}}\phi(\mathcal{D}), 
\ \text{where }
\ \phi(\mathcal{D}):=\max\limits_{q \in \mbox{conv}(\mathcal{Q})}\langle q, \mathcal{P}- \mathcal{D} \rangle \tag{\textcolor{blue}{$P$}}.
\end{equation*}
The following auxiliary Lemma
characterizes this problem as a convex optimization problem.
\begin{lema}
\label{lem:phi_convex}
The function
$\phi:{\Delta}_k\mapsto \mathbb{R}_+$, defined by $\phi(\mathcal{D})$, is convex and 1-Lipschitz with respect to ~$\|\cdot\|_1$. 
\end{lema}

\begin{proof}[Proof of Lemma \ref{lem:phi_convex}]
As $\phi$ is defined as the maximum of linear functions, it is convex. Lipschitzness can be proved by the triangle inequality: For any $\mathcal{D},\mathcal{D}' \in \Delta_k $ it holds that
\[ \phi(\mathcal{D})-\phi(\mathcal{D}') \leq \max_{q\in \text{conv}(\mathcal{Q})} \langle q,\mathcal{D}-\mathcal{D}'\rangle\leq \|\mathcal{D}-\mathcal{D}' \|_1,  \]
where we used the H\"older inequality and that $q\in [-1,1]^k$ for all $q\in \mathcal{Q}$.
\end{proof}

In Section \ref{sec:stoch_comp_min}, the following  approximation to the primal problem \((P_{\alpha})\) is considered:
\begin{equation*} 
\min_{\mathcal{D }\in \Delta_{k}} \Phi_{\alpha,\sigma}(\mathcal{D}): = \min_{\mathcal{D }\in \Delta_{k}} \left[\phi_{\sigma}(\mathcal{D}) + \alpha H(\mathcal{D}) \right],
\tag{$\textcolor{blue}{P_{\alpha,\sigma}}$}
\end{equation*}
where $\phi_{\sigma}(\mathcal{D}):=\int_{\mathbb{R}^{k}}\phi(\mathcal{D}+x)\varphi_{\sigma}(x)\,dx=\mathbb{E}_{\xi} [\phi(\mathcal{D}+\xi)]$
for a Gaussian random vector  $\xi\sim \mathcal{N}(0, \sigma^2 I_k)$.

The following Proposition
characterizes this problem as a convex optimization problem:

\begin{prop}\label{prop:rand_smooth}
Let $\mathcal{Q}\subseteq[-1,1]^k$ be a compact and convex set. Then, the following holds:
\begin{enumerate}
    \item[(i)] $\|\phi_{\sigma}-\phi\|_{\infty}\leq \sigma w(\mathcal{Q})$, where $w(\mathcal{Q}):=\mathbb{E}_{\xi\sim \varphi}\left[\max\limits_{q\in \mathcal{Q}}\langle q,\xi\rangle\right]$ is the Gaussian width of $\mathcal{Q}$ and $\varphi$ denotes the density of the  multivariate standard normal distribution.
    \item[(ii)] $\phi_{\sigma}$ is convex and 1-Lipschitz with respect to~$\|\cdot\|_1$.
    \item[(iii)] $\phi_{\sigma}$ is $1/\sigma$-smooth with respect to $\|\cdot\|_1$.
\end{enumerate}
\end{prop}

\begin{proof}[Proof of Proposition \ref{prop:rand_smooth}]
We prove each part separately:
\begin{enumerate}
    \item[(i)] Let $\mathcal{D}\in\Delta_{ k}$. Then, it holds that 
    \begin{align*}
        \phi_{\sigma}(\mathcal{D})-\phi(\mathcal{D})&=\mathbb{E}_{\xi\sim \varphi_\sigma}\Big[\max_{q\in Q}\langle q, {\mathcal{P}} -\mathcal{D}+\xi\rangle\Big]-\max_{q\in Q}\langle q,\mathcal{P}-\mathcal{D} \rangle \\
        & \leq \mathbb{E}_{\xi \sim \varphi_\sigma}\Big[\max_{q\in Q}\langle q, \xi\rangle \Big]\\
        &=   \sigma w( Q).
    \end{align*}
    \item[(ii)] Convexity of $\phi_{\sigma}$ follows from Lemma 2.1 in \cite{bertsekas1972stochastic} due to convexity of $\phi$. 
Let $\mathcal{D}_1, \mathcal{D}_2 \in \mathbb{R}^k$. Then,  Jensen's inequality and 1-Lipschitz continuity of $\phi$ yield
\begin{align*}
|\phi_\sigma(\mathcal{D}_1) - \phi_\sigma(\mathcal{D}_2)| =& \left|\mathbb{E}_{\xi \sim \varphi_{\sigma}}\left[\phi(\mathcal{D}_1 + \xi)\right] - \mathbb{E}_{\xi \sim \varphi_{\sigma}}\left[\phi(\mathcal{D}_2 +  \xi)\right]\right|\\
\leq &\mathbb{E}_{\xi \sim \varphi}\left[|\phi(\mathcal{D}_1 + \xi) - \phi(\mathcal{D}_2 + \xi)|\right]\\
\leq &\|\mathcal{D}_1 - \mathcal{D}_2\|_1.
\end{align*}
Therefore, $\phi_\sigma$ is 1-Lipschitz with respect to the $\|\cdot\|_1$-norm.
    \item[(iii)] 
    Let $\xi=(\xi_1, \ldots, \xi_{k})^{\intercal}$ be a multivariate normal random vector with mean $\mathbf{0}$ and covariance matrix $\sigma^2I_{k}$.
According to  Lemma 1.5 in \cite{Abernethy:2016}
    $\nabla \phi_{\sigma}(\mathcal{D})=\mathbb{E}_{\xi}[\phi(\mathcal{D}+\xi)\nabla \nu(\xi)]$, where  $\nu(\xi)=\frac{\|\xi\|_2^2}{2\sigma^2}$. It then follows that
    \begin{align*}
        \|\nabla\phi_{\sigma}(\mathcal{D}_1)-\nabla\phi_{\sigma}(\mathcal{D}_2)\|_{\infty}
        &=\|\mathbb{E}_{\xi}\big[\big(\phi(\mathcal{D}_1+\xi)-\phi(\mathcal{D}_2\xi)\big)\nabla \nu(\xi)\big] \|_{\infty}\\
        &= \frac{1}{\sigma^2}\max_{j\in[ k]}\big|\mathbb{E}_{\xi}\big[\big(\phi(\mathcal{D}+\xi)-\phi(\mathcal{D}^{\prime}+\xi)\big)\xi_i\big]\big| \\
        &\leq \sup_{x\in \mathbb{R}^{k}}|\phi(\mathcal{D}_1+x)-\phi(\mathcal{D}_2+x)|\cdot \max_{j\in[ k]}\frac{\mathbb{E}_{\xi}|\xi_j|}{\sigma^2},
    \end{align*}
    where in the last step we used  $1$-Lipschitzness of $\phi$, and the Jensen inequality.
\end{enumerate}
\end{proof}

\section{Proofs for the results in  Section \ref{sec:sol_dual}}
\label{sec:supp_A}

This section establishes proofs for Theorems \ref{thm:utility_guarantee} and \ref{thm:upper_bound_primalgap} in Section \ref{sec:sol_dual} of the main document.  Section
\ref{sec:report_noisy_max}
establishes the Laplace Mechanism as a tool for the incorporation of differential privacy in algorithms for synthetic data release and reviews some key properties of differentially private algorithms.
Finally, in Sections \ref{sec:proof_thm_1}  and \ref{proof-thm-41}, we provide the proofs for Theorem \ref{thm:utility_guarantee} and \ref{thm:upper_bound_primalgap}, respectively.

%Both  $\mathcal{P}_\alpha^{\text{priv}}$ and  $\D_\alpha^\ast$ exist and are unique because $H$ is strongly convex. 

\subsection{The Laplacian mechanism and key properties of differential privacy}\label{sec:report_noisy_max}
A common method for ensuring differential privacy involves adding noise to the output of a deterministic function $f$. In Algorithm \ref{alg:DPFW}, we apply such a noise-adding technique, known as the \textit{Laplace Mechanism}, which guarantees pure differential privacy. To maintain generality, we present this mechanism in the context of a generic function $f$ and dataset $S$ (for further details, see \cite{dwork2006calibrating}). Specifically, let $\varepsilon \geq 0$ and $S$ represent a dataset. For $k \geq 1$, let $f(S)\in \RR^k$ be a deterministic function of $S$. The  following (randomized) algorithm $\mathcal{A}$ is known as the \textit{Laplace Mechanism}:
\begin{equation}
\mathcal{A}(S)= f(S)+ (Y_1, \ldots, Y_k)^\intercal\;,    \label{def:Laplace_mechanism}
\end{equation}
where $Y_1, \ldots, Y_k\overset{i.i.d}{\sim} \text{Lap}\left( \frac{\Delta(f)}{\varepsilon}\right)$ and where the quantity $ \Delta(f):= \sup_{S_1 \sim S_2 \\ S_1, S_2 \in \mathcal{Z}^n} \| f(
S_1)-f(S_2) \|_{\ell^1}$ is called \textit{sensitivity of } $f$. The supremum is taken over all neighboring datasets.
The intuition behind the Laplace Mechanism is to add noise drawn from a Laplace distribution to the output of an algorithm, with the amount of noise calibrated to the algorithm's sensitivity (i.e., how much its output changes in response to a single data point). By introducing this scaled noise, the mechanism ensures that the output remains similar whether or not a particular individual is in the dataset, thus providing differential privacy. It can be shown that $\mathcal{A}$ in \eqref{def:Laplace_mechanism} is $(\varepsilon, 0)$\, -\, ; see \cite{dwork2006calibrating}. 
%Due to the post-processing property of differential privacy, we also obtain privacy guarantees for the following procedure, which will be utilized later: 
Suppose that $f(S) = (f_1(S), \ldots, f_k(S))^\intercal$, where each $f_j(S) \in \mathbb{R}$ is a function of the dataset $S$, and  for all $j=1,\ldots, k$  the sensitivity is of the form $\Delta(f_j) \leq L$ for some constant $L$. Then, the   so-called Report Noisy Max (RNM) variable, defined by
\begin{equation}
    \text{RNM}(S) = \argmax\limits_{i=1, \ldots, k} \left\{ f_i(S) + Y_i \right\}\;, 
    \label{eq:RNM}
\end{equation}
where $ Y_j \overset{i.i.d}{\sim} \text{Lap}\left(\frac{L}{\varepsilon}\right)$,  is $(\varepsilon, 0)$-DP; see 
Claim 3.9 in \cite{dwork2014algorithmic}.

While the Laplace mechanism introduces differential privacy into the Differentially private Frank-Wolfe algorithm (Algorithm \ref{alg:DPFW}), the iterative nature of the algorithm raises the question how much
privacy is lost due to composition of the individual iterations. 
For Algorithm 1, 
an answer to this question  will be provided in Section \ref{sec:proof_thm_1}. It will be based on the well-known
 Advanced Composition Theorem:
\begin{teo}
[Theorem III.3, \cite{5670947}] For all $\varepsilon>0$ and $ 1> \delta>0 $ with $\log(1/\delta) \geq \varepsilon^2 T$, let $\mathcal{A}=(\mathcal{A}_1,\ldots, \mathcal{A}_T)$ be a sequence of $(\varepsilon, 0)-DP$ algorithms where each $\mathcal{A}_j $ is sequentially and adaptively chosen. Then, the whole chain $\mathcal{A}$ is $(\hat{\varepsilon}, \delta)$-DP, where $
     \hat{\varepsilon}= 4\varepsilon \sqrt{2T \log(1/\delta)}\;. $
     \label{thm:advance_composition_theorem}
\end{teo}

We conclude this section by presenting another key property of Differential Privacy we will rely on: the \emph{post-processing}: If $\mathcal{A}$ is $(\varepsilon, \delta)$-DP and $f: \mathcal{O} \longrightarrow \mathcal{O}'$ is any deterministic function, then the composition $f \circ \mathcal{A}: \mathcal{Z}^n \longrightarrow \mathcal{O}'$ is also $(\varepsilon, \delta)$-DP. In other words, no matter what operations or transformations are applied to the outputs of a differentially private algorithm, the privacy guarantees provided by $\mathcal{A}$ will still hold; see Proposition 2.1 in \cite{dwork2014algorithmic}.

\subsection{Proof of Theorem 1: Differential privacy and accuracy of Algorithm \ref{alg:DPFW}}\label{sec:proof_thm_1}

First, we establish differential privacy of the algorithm. For this, let $S_n=\{z_1, \ldots, z_n\}\overset{i.i.d}{\sim}\mathcal{P}$ and $\mathcal{Q}=\{ q^{(1)}, \ldots, q^{|\mathcal{Q}|}\}$. Recall that $q_{t+1}=q_t+\gamma(s_t-q_t)$, where $s_t$ is defined by line 9 of Algorithm \ref{alg:DPFW}. Within each iteration of the algorithm, we use  Report Noisy Max on the vector $\left(  \langle \hat{\nabla}\psi_\alpha(q_t), q^{(1)} \rangle, \ldots,  \langle \hat{\nabla}\psi_\alpha(q_t), q^{(|\mathcal{Q}|)} \rangle \right)^\intercal $, i.e.
$$\mathcal{A}(S_n):= \left( \begin{array}{c}
     \langle \hat{\nabla}\psi_\alpha(q_t), q^{(1)} \rangle  \\
     \vdots \\
      \langle \hat{\nabla}\psi_\alpha(q_t), q^{(|\mathcal{Q}|)} \rangle
\end{array}\right)+ \left( \begin{array}{c}
     u_{1,t}  \\
     \vdots \\
     u_{|\mathcal{Q}|,t}
\end{array}\right)\; \quad \text{with}\quad u_{i,t}\overset{i.i.d}{\sim}\text{Lap}\left( \lambda \right)\;, $$
where $\hat{\nabla}\psi_\alpha(q_t) =\mathcal{P}_n  -\nabla H^\ast \left( \frac{q_t}{\alpha}\right)=\frac{1}{n}\sum_{i=1}^n e_{z_i}-\nabla H^\ast \left( \frac{q_t}{\alpha}\right)\;.$ For a neighboring dataset $S_n'=\{z'_1,\ldots,z'_n\}$ of $S_n$, the corresponding $ \hat{\nabla}\psi_\alpha'(q_t)= \frac{1}{n}\sum_{i=1}^n e_{z'_i}  -\nabla H^\ast \left( \frac{q_t}{\alpha}\right)=\mathcal{P}'_n  -\nabla H^\ast \left( \frac{q_t}{\alpha}\right).$
Due to Claim 3.9 in \cite{dwork2014algorithmic} discussed in Section \ref{sec:report_noisy_max}, for the algorithm to be $(\varepsilon, 0)$-DP,
the parameter $\lambda$ must correspond to the maximum sensitivity of a single iteration within the algorithm. We consider the $i-$th iteration, i.e. the $i$-th entry of $\mathcal{A}(S_n)$, and define its deterministic part $f_i(S_n):=  \langle q^{(i)}, \hat{\nabla}\psi_\alpha(q_t) \rangle $, while 
for a neighboring set $S_n'$, we define
$f_i(S'_n):=  \langle q^{(i)}, \hat{\nabla}\psi'_\alpha(q_t) \rangle$. The corresponding sensitivity is
\begin{align*}
    \Delta (f_i(S_n))=& \sup_{S_n \sim S_n'} \left| \langle q^{(i)}, \hat{\nabla}\psi'_\alpha(q_t) - \hat{\nabla}\psi_\alpha(q_t) \rangle\right| \\
    \leq& D_1 \sup_{S_n \sim S_n'} \left\| \frac{1}{n}\sum_{z_i \in S_n} e_{z_i} - \nabla H^\ast \left( \frac{q_t}{\alpha}\right)-  \frac{1}{n}\sum_{z'_i \in S'_n} e_{z'_i} + \nabla H^\ast \left( \frac{q_t}{\alpha}\right) \right\|_\infty =  \frac{D_1}{n}=:L\;. 
\end{align*}
 As a result, since the sensitivity of each entry $f_j$ is upper bounded by a universal constant $L$, independent of $j$, if we want the Algorithm to be $(\tilde{\varepsilon},0)\,-\,$DP, we need  $\lambda=  \frac{D_1}{n\tilde{\varepsilon}}$.
Since it performs $T$ iterations, the Advanced Composition Theorem, i.e. Theorem \ref{thm:advance_composition_theorem} in Section \ref{sec:report_noisy_max}, yields a privacy guarantee of
\[ \left(  4\tilde{\varepsilon} \sqrt{2T \log \left( \frac{1}{\delta}\right)} ,\delta \right) \text{- DP}\;.\]
Hence, in order to make the output $(\varepsilon,\delta)$-DP, $\tilde{\varepsilon}$ needs to be chosen as
\( \tilde{\varepsilon} = \varepsilon\left(4\sqrt{2T \log \left( \frac{1}{\delta}\right)} \right)^{-1}.\)
This concludes the privacy analysis.

In what follows, we establish  the accuracy bound provided by Theorem \ref{thm:utility_guarantee}.
To simplify the analysis and follow the standard proof structure for the convergence of the Frank-Wolfe algorithm, we focus on \(\mathcal{J}(q) := -\psi_\alpha(q)\), which is strongly convex. Note that this means that we replace line 9 of Algorithm \ref{alg:DPFW} with 
\begin{align*}
s_t  = \argmin\limits_{s \in \mbox{conv}({\cal Q})} \{ \langle \hat{\nabla}\mathcal{J}(q_t), s \rangle + u_{s,t} \}.
\end{align*}
(For this, note that $u_{s,t}\overset{\D}{=}-u_{s,t}$ by symmetry of the centered Laplace distribution.) 

Our goal is to bound the Frank-Wolfe gap
\begin{align*}
g_\mathcal{J}(q) & := \max_{s \in \mbox{conv}({\cal Q})}\langle \nabla \mathcal{J}(q), q-s \rangle.
\end{align*}
For this, we leverage the descent lemma (which holds due to the \( L_1 \)-smoothness of \( \psi_\alpha \))

\begin{align*}
\mathcal{J}(q_{t+1}) \leq & \mathcal{J}(q_t) + \langle \nabla \mathcal{J}(q_t), q_{t+1}-q_t\rangle + \frac{1}{2\alpha} \|q_{t+1}-q_t\|^2_\infty \quad \text{($\mathcal{J}$ is $\frac{1}{\alpha}-$smooth for Lemma \ref{smoothness_dual}})\\
\leq & \mathcal{J}(q_t) + \gamma \langle \nabla \mathcal{J}(q_t) , s_t -q_t\rangle + \frac{\gamma^2 D_\infty^2}{2\alpha}  \quad \text{(Definition of }q_{t+1}\;,\, \text{ line 10 of Algorithm }\ref{alg:DPFW}) \\
= & \mathcal{J}(q_t) + \gamma \langle \nabla \mathcal{J}(q_t) , s_t -p_t \rangle + \gamma \langle \nabla \mathcal{J}(q_t) , p_t -q_t\rangle + \frac{\gamma^2 D_\infty^2}{2\alpha}\;\\
= & \mathcal{J}(q_t) + \gamma \langle \nabla \mathcal{J}(q_t) , s_t -p_t \rangle - \gamma g_\mathcal{J}(q_t ) + \frac{\gamma^2 D_\infty^2}{2\alpha}\;,
\end{align*}
where 
\begin{align*}
    p_t & =\argmin \limits_{s \in \mbox{conv}({\cal Q})} \langle \nabla \mathcal{J}(q_t),s - q_t \rangle\;. \\
\end{align*}
%In the last inequality we added and subtracted the term $p_t$ in order to make the Frank-Wolfe gap $g_\mathcal{L}$ appear. In fact, by definition,
%\begin{align*}
%g_\mathcal{J}(q) & := \max_{s \in \mbox{conv}({\cal Q})}\langle \nabla \mathcal{J}(q), q-s \rangle \quad \text{implies} \quad g_\mathcal{J}(q_t ) = \langle \nabla \mathcal{J}(q_t ), q_t -p_t \rangle\;.
%\end{align*}
As a next step, we introduce the auxiliary variable \( V_t \)  defined as follows:
\begin{align*}
V_t & := \max_{{s \in \mathcal{Q}}} u_{s,t} - \min_{{s \in \mathcal{Q}} } u_{s,t}.
\end{align*}
%Later, we will need to use a technicality (inequality \eqref{eq-st-choice}), which arises from the following observation: 
Since, by  definition of \( s_t \), for all \( q \in \text{conv}(\mathcal{Q}) \), for all \( t = 1, \dots, T \), and for all \( s \in \mathcal{Q} \)
\[  \langle \hat{\nabla}\mathcal{J}(q),s_t \rangle + \min_{s \in \mathcal{Q}} u_{s,t} \leq 
 \langle \hat{\nabla}\mathcal{J}(q),s_t \rangle +  u_{s,t}  \leq 
\langle \hat{\nabla}\mathcal{J}(q) ,s \rangle + \max_{s \in \mathcal{Q}} u_{s,t}  \;, \] 
it holds that 
\begin{equation}
\label{eq-st-choice}
\begin{aligned}
 \quad \langle \hat{\nabla}\mathcal{J}(q),s_t -s \rangle  & \leq V_t\;  \ \text{a.s.} \ \ \forall s\in \mbox{conv}({\cal Q}).
\end{aligned}
\end{equation}
%We now have all the tools necessary to bound the Frank-Wolfe gap, \( g_{\mathcal{J}} \). 
It therefore follows that 
\begin{align*}
\mathcal{J}(q_{t+1}) \leq & \mathcal{J}(q_t)  - \gamma g_\mathcal{J}(q_t) +\gamma \langle \nabla \mathcal{J}(q_t) , s_t- p_t \rangle + \frac{\gamma^2 D_\infty^2}{2\alpha} \\
= & \mathcal{J}(q_t)  - \gamma g_\mathcal{J}(q_t) +\gamma \langle \nabla \mathcal{J}(q_t) - \hat{\nabla}\mathcal{J}(q_t) , s_t- p_t \rangle + \gamma \langle \hat{\nabla}\mathcal{J}(q_t), s_t - p_t \rangle + \frac{\gamma^2 D_\infty^2}{2\alpha} \\
\leq &  \mathcal{J}(q_t)  -\gamma g_\mathcal{J}(q_t) +\gamma \langle \nabla \mathcal{J}(q_t)-\hat{\nabla}\mathcal{J}(q_t), s_t- p_t \rangle + \gamma V_t+ \frac{\gamma^2 D_\infty^2}{2 \alpha}  .
\end{align*}

Rearranging the inequality yields

  \begin{align*}
g_\mathcal{J}(q_t)\leq &
  \frac{1}{\gamma}(\mathcal{J}(q_t) - \mathcal{J}(q_{t+1})) +\Big( \| \nabla \mathcal{J}(q_t)-\hat{\nabla}\mathcal{J}(q_t) \|_{\infty} \| s_t - p_t \|_1 \Big)+  V_t+  \frac{\gamma D_\infty^2  }{2\alpha},
  \end{align*}
  which implies
\begin{align*}
\sum_{t=0} ^{T-1} g_\mathcal{J}(q_t)\leq &
  \frac{1}{\gamma} \sum_{t=0} ^{T-1} (\mathcal{J}(q_t) - \mathcal{J}(q_{t+1}))+ D_1 \sum_{t=0}^{T-1}\| \nabla \mathcal{J}(q_t)-\hat{\nabla}\mathcal{J}(q_t) \|_{\infty}+ \sum_{t=0}^{T-1}V_t+  \frac{\gamma D_\infty^2 T}{2\alpha} \;\\
  = &
  \frac{1}{\gamma} \mathcal{J}(q_0)-\mathcal{J}(q_T)+ D_1 \sum_{t=0}^{T-1}\| \nabla \mathcal{J}(q_t)-\hat{\nabla}\mathcal{J}(q_t) \|_{\infty}+ \sum_{t=0}^{T-1}V_t+  \frac{\gamma D_\infty^2 T}{2\alpha} \;. \\
\end{align*}
Moreover, note that
\begin{align*}
\|\hat{\nabla}\mathcal{J}(q_t)-\nabla \mathcal{J}(q_t)\|_\infty=&\left \| \frac{1}{n}\sum_{i=1}^n e_{z_i} - \nabla H^\ast\left(\frac{q_t}{\alpha} \right) - \EE_{\mathbf{z}\sim \mathcal{P}}[e_\mathbf{z}] +\nabla H^\ast\left(\frac{q_t}{\alpha} \right)\right\|_\infty =\left\| \frac{1}{n}\sum_{i=1}^n e_{z_i}- \EE_{\mathbf{z}\sim \mathcal{P}}[e_\mathbf{z}] \right\|_\infty \;.
\end{align*}
 Further, since $\cal J$ is $2$-Lipschitz due to Lemma \ref{smoothness_dual},  $
{\cal J}(q_0) - {\cal J}(q_t) \leq 2 \| q_0 -q_t \| \leq 2 D_\infty  $, and we can simplify the above expression  to
\begin{align*}
   \frac{1}{T}\sum_{t=0} ^{T-1} g_\mathcal{J}(q_t)\leq &
      \frac{2 D_\infty }{\gamma T}+ D_1 \left\| \frac{1}{n}\sum_{i=1}^n e_{z_i}- \EE_{\mathbf{z}\sim \mathcal{P}}[e_\mathbf{z}] \right\|_\infty + \frac{1}{T}\sum_{t=0}^{T-1}V_t+  \frac{D^2_\infty \gamma}{2\alpha}\;. 
\end{align*}
Considering the expectation of the previous expression, we obtain
\begin{align*}
 \EE \left[ \frac{1}{T}\sum_{t=0} ^{T-1} g_\mathcal{J}(q_t) \right]\leq   \frac{2 D_\infty }{\gamma T} +  D_1 \EE\Big\| \frac{1}{n}\sum_{i=1}^n e_{z_i}- \EE_{\mathbf{z}\sim \mathcal{P}}[e_\mathbf{z}] \Big\|_\infty+ \frac{1}{T}\sum_{t=0}^{T-1}\EE [V_t] +  \frac{1}{2\alpha} D^2_\infty \gamma  \;. 
\end{align*}
It follows from Corollary 2.3 in \cite{PMID:20224805}
that  $\EE\left[\| \frac{1}{n}\sum_{i=1}^n e_{z_i}- \EE_{\mathbf{z}\sim \mathcal{P}}[e_\mathbf{z}] \|_\infty \right]  \leq \mathcal{O} \left( \sqrt{\log k/n }\right)$. Additionally, according to Lemma \ref{lemma:properties_laplacians}, for all $t = 0, \ldots, T-1$, it holds that $\EE[V_t] \leq 2 \lambda \log(2|{\cal Q}|)$. We conclude that
\begin{align*}
\EE \left[ \frac{1}{T}\sum_{t=0} ^{T-1} g_\mathcal{J}(q_t) \right]\leq 
\frac{2 D_\infty }{\gamma T} +  D_1 {\cal O} \left( \sqrt{\frac{\log k}{n}} \right)  +2 \lambda \log 2 |{\cal Q}|+  \frac{1}{2\alpha} D^2_\infty \gamma  \;.
\end{align*}

Recall that $\lambda = \frac{D_1}{n\tilde{\varepsilon}}=\frac{4\sqrt{2} \sqrt{T \log(1/\delta)}}{\varepsilon n}$. Minimizing the right hand side of the above inequality with respect to $\gamma$, we find that the minimizer $\gamma^\ast = 2 \sqrt{\frac{\alpha}{T D_\infty }}$ yields
\begin{align*}
\EE \left[ \frac{1}{T}\sum_{t=0} ^{T-1} g_\mathcal{J}(q_t) \right]\leq &
2 D_\infty^{3/2} \sqrt{\frac{1}{\alpha T}} +  D_1 {\cal O} \left( \sqrt{\frac{\log k}{n}} \right)   + \frac{8\sqrt{2} \log(2|{\cal Q}| ) \sqrt{T \log(1/\delta)}}{\varepsilon n} \;.
\end{align*}

Finally, minimizing the right-hand side of the above inequality with respect to $T$ yields $T^{\ast} = \frac{ D_\infty^{3/2} \varepsilon n}{\sqrt{32\alpha \log(1/\delta) }\log (2 |\mathcal{Q}|)}$ as minimizer and 
\[  \EE[g_ {\mathcal{L}}(q_\alpha^\text{out}))] = \mathcal{O}\left( \frac{D_\infty^{3/4} \log^{1/4} (1/\delta) \log^{1/2}(|\mathcal{Q}|)}{\alpha^{1/4} (\varepsilon n)^{1/2}}  +\frac{ D_1 \log^{1/2}(k)}{\sqrt{n}}\right) \;,\]
where $q_\alpha^{\text{out}}$ denotes the output of the Frank-Wolfe Algorithm.

\subsection{Proof of Theorem \ref{thm:upper_bound_primalgap}: Upper bound on primal gap}
\label{proof-thm-41}

This section is structured as follows. We begin by proving a key technical result (Lemma \ref{KKT lemma}), which forms the foundation for the proof of Theorem \ref{thm:upper_bound_primalgap}. Theorem \ref{thm:upper_bound_primalgap}  then establishes an upper bound for 
the accuracy of 
\begin{align*}
&\mathcal{P}_\alpha^{\text{priv}}=\argmin\limits_{ \mathcal{D} \in \Delta_{k}} \left( \langle q_\alpha^{\text{out}}, \mathcal{P} - \mathcal{D}\rangle+\alpha H(\mathcal{D}) \right)\;.\\
\end{align*}
Lastly, we conclude by presenting the proof of Theorem \ref{MAINRESULTS}.

\begin{lema} \label{proof-close-form}
Consider the negative entropy function $H$ on $\Delta_k$, $H(\mathcal{D})=\sum_{1 \leq i \leq k} \mathcal{D}(i) \log (\mathcal{D}(i) )$,  and the Bregman divergence $D_{H}(\mathcal{D}, \mathcal{D}')$. Moreover, suppose that $ \mathcal{D}'\in \Delta_k $ has only strictly positive coefficients. For $A\in \mathbb{R}$, $B > 0$, $C \geq 0 $ and vectors $g \in \mathbb{R}^{k}$, the following problem 
\begin{align*}
\min_{\mathcal{D} \in {\Delta_k}} \Psi_{A,B,C}(\mathcal{D}), \text{where} \ \Psi_{A,B,C}(\mathcal{D}): =  A \langle g  , \mathcal{D} \rangle + B H(\mathcal{D}) + C D_{H}(\mathcal{D},\mathcal{D}') 
\end{align*}
has an unique solution $\mathcal{D}_\ast \in \Delta_k $ and 
\[ \mathcal{D}_\ast \propto \exp \Big( \frac{- A g_{i} - B + C \log \mathcal{D}'(i)}{B+C} \Big).
 \]
\end{lema}

\begin{proof}[Proof of Lemma \ref{proof-close-form}]
Note that the first summand in the representation of $\Psi_{A,B, C}$ is linear, while the second summand corresponds to the negative entropy function which is strongly convex according to Example 2.5 in \cite{shalev2012online}.
Moreover,  recall that Lemma  \ref{lem:strong_conv_div} establishes strong convexity   $D_{H}(\cdot ,D')$. As a result,   $\Psi_{A,B,C}$ is strongly convex, as well. Therefore, there is only an unique minimum  $\mathcal{D}$ in $(\mathbb{R}_{+})^{k}$ which is also the stationary point i.e. 
\begin{equation}
\label{stationnary-mirror-descent}
\nabla \Psi_{A,B,C}(\mathcal{D}) = A  g + B \nabla H(\mathcal{D}) + C [\nabla H(\mathcal{D}) - \nabla H(\mathcal{D}') ] = 0.
\end{equation}  
We evaluate  equation \eqref{stationnary-mirror-descent} in its vector form. For all $1 \leq i \leq k$,
\begin{align*}
A g_i + (B+C ) \Big[ 1 + \log \mathcal{D}(i) \Big] - C \Big( 1 + \log \mathcal{D}'(i) \Big)  = 0 ,  
\end{align*}
which yields 
\begin{align*}
(B+C) \log \mathcal{D}(i) & = - A g_{i} - B + C \log \mathcal{D}'(i) \\
\mathcal{D}(i) & = \exp \left( \frac{- A g_{i} - B + C \log \mathcal{D}'(i)}{B+C} \right).
\end{align*}
For the optimality in $\Delta_k$, we will study $\Psi(\lambda \mathcal{P})$ for any $\lambda \in \mathbb{R}, \mathcal{P} \in (\mathbb{R}_{+})^{k}$. We notice that $H(\lambda \mathcal{P}) = \lambda H(\mathcal{P}) + \lambda \log(\lambda) \sum_i \mathcal{P}(i)$, therefore 
\begin{align*}
\Psi_{A,B,C}(\lambda \mathcal{P}) & =   A\langle g, \lambda \mathcal{P} \rangle  + B H(\lambda \mathcal{P})  + C [ H(\lambda \mathcal{P})  - H(\mathcal{D}') -  \langle \nabla H(\mathcal{D}'), \lambda \mathcal{P} - \mathcal{D}' \rangle ] \\
& =  \lambda \Big[A \langle g,\mathcal{P} \rangle + B H(\mathcal{P})  + C [ H( \mathcal{P})  - H(\mathcal{D}') -  \langle \nabla H(\mathcal{D}'), \mathcal{P} - \mathcal{D}' \rangle ] \Big] \\
& + B \lambda \log(\lambda) \sum_i \mathcal{P}(i) + C \lambda \log(\lambda) \sum_i \mathcal{P}(i) + C (\lambda -1) H(\mathcal{D}')  - C (\lambda - 1) \langle \nabla H(\mathcal{D}'),  \mathcal{D}' \rangle \\
& = \lambda \Psi_{A,B,C}(\mathcal{P}) + ( B+C) \lambda \log(\lambda) \Big(\sum_i \mathcal{P}(i) \Big)  +  C (\lambda -1)  \Big( H(\mathcal{D}') - \langle \nabla H(\mathcal{D}'),  \mathcal{D}' \rangle \Big)
\end{align*}
We notice that the second term scales with $\sum_i \mathcal{P}(i)$ and the last one is independent of $\mathcal{P}$. Now we consider $\mathcal{D}_\ast \in \Delta_k $
\[ \mathcal{D}_{\ast}(i):= \frac{\mathcal{D}(i) }{\sum_{j=1}^{k} \mathcal{D}(j)}, \quad  \lambda_\ast : = \frac{1}{\sum_{j=1}^{k} \mathcal{D}(j)} \]
We know that for all $\mathcal{P} \in \frac{1}{\lambda_\ast} \Delta_k$ (i.e. $\lambda_\ast \mathcal{P} \in \Delta_k $)
\begin{align*}
\Psi_{A,B,C}(\lambda_\ast \mathcal{P}) - \Psi_{A,B,C}(\mathcal{D}_\ast) & = \Psi_{A,B,C}(\lambda_\ast \mathcal{P}) - \Psi_{A,B,C}(\lambda_\ast \mathcal{D}) \\
& = \lambda_\ast \Big (\Psi_{A,B,C}(\mathcal{P}) - \Psi_{A,B,C}(\mathcal{D}) \Big) \geq 0 
\end{align*}
with equality only if $\mathcal{P}$ = $\mathcal{D}$.
\end{proof}

\begin{cor} 
 \label{KKT lemma} 
Given the assumptions of Theorem \ref{proof-close-form}, consider
 \begin{align}
     p= \argmin\limits_{ \mathcal{D} \in \Delta_{k}} \left[\langle q, - \mathcal{D} \rangle + \alpha H(\mathcal{D}) \right] \;.
     \label{eq:optimal_p}
 \end{align} 
 Then, it holds that
 \begin{enumerate}
     \item $\langle q, p\rangle = \alpha H(p) +\alpha H^\ast \left( \frac{q}{\alpha}\right)\;,$
     \item $p=\nabla H^\ast \left( \frac{q}{\alpha}\right)\;. $
 \end{enumerate}
\end{cor}

 \begin{proof}[Proof of Corollary \ref{KKT lemma}] 
 \begin{enumerate}
     \item  Note that \begin{align*}
    \min\limits_{ \mathcal{D} \in \Delta_{k}} \left[\langle q, - \mathcal{D} \rangle + \alpha H(\mathcal{D}) \right] = -\alpha \max\limits_{ \mathcal{D} \in \Delta_{k}} \left[\langle \frac{q}{\alpha},  \mathcal{D} \rangle - H(\mathcal{D}) \right]=-\alpha  H^\ast \left( \frac{q}{\alpha}\right)\;.
\end{align*}
By definition, the argmin \( p=(p_1, \ldots, p_k) \) is the value that attains the minimum, and due to strong convexity, this value is unique. 
\item Using the properties of the Fenchel conjugate, for all $\mathcal{D}\in \Delta_{k}$ it holds  $\mathcal{D}= \nabla H^\ast \left(  \nabla H (\mathcal{D})\right)\;.$ Therefore, we only have to compute $\nabla H$, evaluated at the optimum $p$. Since $p$ is by definition the optimum, it must satisfy the  Karush–Kuhn–Tucker (KKT) conditions for \eqref{eq:optimal_p}. 
Recall that $\mathcal{D}=(\D(1), \D(2), \ldots, \D(k))$.
The Lagrangian associated to \eqref{eq:optimal_p} is
$$L(\mathcal{D}, \lambda_1, \ldots, \lambda_{k}, \mu)= -\sum_{j=1}^{k} q_j \mathcal{D}(j) + \alpha H(\mathcal{D})+  \mu \left(  1- \sum_{j=1}^{k} \mathcal{D}(j) \right)  -\sum_{j=1}^{k}\lambda_j \mathcal{D}(j)\;.$$

From Lemma \ref{proof-close-form}, for all $j=1, \ldots, k$, $p_j$ are strictly positive and $\lambda_j=0$). Therefore, $\nabla_\mathcal{D}L(\mathcal{D}, \mu):=\nabla_\mathcal{D}L(\mathcal{D}, 0,\ldots, 0, \mu)$ evaluated at the optimum $p$, 
\begin{align*}
    \nabla_\D L (p, \mu)= -q +\alpha \nabla H(p) -\mu \mathbf{1}=0\;,
\end{align*}
where $\mathbf{1}=[1,\ldots, 1]$.  Therefore, $\nabla H(p)= \frac{q}{\alpha } + \frac{\mu \mathbf{1}}{\alpha}\;,$ and $p=\nabla H^\ast \left(  \nabla H (p)\right) = \nabla H^\ast \left( \frac{q}{\alpha } + \frac{\mu \mathbf{1}}{\alpha}  \right)\;.$ The last step is to show that $\nabla H^\ast \left( \frac{q}{\alpha } + \frac{\mu \mathbf{1}}{\alpha}  \right)= \nabla H^\ast \left( \frac{q}{\alpha }  \right)\;.$ The Fenchel conjugate of the negative entropy is 
$H^\ast(y_1, \ldots, y_{k}) = \log \left( \sum_{j=1}^{k} e^{y_j}\right)\;,$ and its gradient $\nabla H^\ast(y_1, \ldots, y_{k})=\left( \sum_{j=1}^{k} e^{y_j}\right)^{-1}  \left(e^{y_1}, \ldots, e^{y_{k}} \right)^\intercal\;.$ Evaluating at  $\frac{q}{\alpha } + \frac{\mu \mathbf{1}}{\alpha}$ we get  $\nabla H^\ast(\frac{q_1}{\alpha}+\frac{\mu}{\alpha}, \ldots, \frac{q_{k}}{\alpha}+\frac{\mu}{\alpha})=\left( \sum_{j=1}^{k} e^{q_j/\alpha}\right)^{-1} e^{-\mu/\alpha}\cdot e^{+\mu/\alpha} \left(e^{q_1/\alpha }, \ldots, e^{q_{k}/\alpha} \right)^\intercal= \nabla H^\ast \left( \frac{q}{\alpha }  \right)\;.$
 \end{enumerate}
 \end{proof}

\begin{proof}[Proof of Theorem \ref{thm:upper_bound_primalgap}]
 By definition  $\textbf{Gap}_{(P)} \left( \mathcal{P}_\alpha^{\text{priv}} \right):= \phi(\mathcal{P}_\alpha^{\text{priv}})-\phi^\ast$, where $\phi^\ast=\min\limits_{\mathcal{D}\in \Delta_{k}} \phi(\mathcal{D})\;.$ Note that for all $\mathcal{D}\in \Delta_{k}$, it holds that $\phi_\alpha(\mathcal{D})=\phi(\mathcal{D})+\alpha H(\mathcal{D})$, and since $H(\mathcal{D})\leq 0$, $\phi^\ast_\alpha=\min\limits_{\mathcal{D}\in \Delta_{k}} \phi_\alpha (\mathcal{D}) \leq \min\limits_{\mathcal{D}\in \Delta_{k}} \phi (\mathcal{D})=\phi^\ast\;.$ Therefore, 
\begin{align*}
    \textbf{Gap}_{(P)}(\mathcal{P}_\alpha^{\text{priv}}) = & \phi(\mathcal{P}_\alpha^{\text{priv}})-\phi^\ast \leq \phi_\alpha(\mathcal{P}_\alpha^{\text{priv}})-\alpha H(\mathcal{P}_\alpha^{\text{priv}}) -\phi^\ast_\alpha \;.
\end{align*}
Next, the positive entropy $-H(\mathcal{D})$ attains its maximal value $\log (k)$ for $\mathcal{D}$  with  $\mathcal{D}(z)=\frac{1}{k}\, \text{for all }z\in\mathcal{Z}\;.$

It  follows that
\begin{align*}
\textbf{Gap}_{(P)}(\mathcal{P}_\alpha^{\text{priv}})  \leq &\phi_\alpha(\mathcal{P}_\alpha^{\text{priv}}) -\phi^\ast_\alpha -\alpha H(\mathcal{P}_\alpha^{\text{priv}})\\
\leq & \phi_\alpha(\mathcal{P}_\alpha^{\text{priv}})-\phi_\alpha^\ast  +\alpha \log (k)  \quad  \\
= &  \phi_\alpha(\mathcal{P}_\alpha^{\text{priv}}) -\psi_\alpha^\ast + \alpha \log(k) \quad \text{(strong duality)}\\
= &  \phi_\alpha(\mathcal{P}_\alpha^{\text{priv}})- \psi_\alpha(q_\alpha^{\text{out}}) + \psi_\alpha(q_\alpha^{\text{out}}) -\psi_\alpha^\ast + \alpha \log(k) \\
 \leq &  \phi_\alpha(\mathcal{P}_\alpha^{\text{priv}})- \psi_\alpha(q_\alpha^{\text{out}}) + \alpha \log(k) \;.\\
 \end{align*}
It remains to deal with the term $\phi_\alpha(\mathcal{P}_\alpha^{\text{priv}})- \psi_\alpha(q_\alpha^{\text{out}})$. Note that we have analytical expressions for both $\phi_\alpha(\cdot)$ and $\psi_\alpha (\cdot)\;.$ For all $\D \in \Delta_{k}$ and $q\in \text{conv}(\mathcal{Q})$,
 \begin{align*}
     \phi_\alpha(\D)= \langle \tilde{q}_{\D}, \mathcal{P}-\mathcal{D}\rangle+\alpha H(\mathcal{D})\quad \text{and}\quad \psi_\alpha(q)=\langle q, \mathcal{P}\rangle-\alpha H^\ast \left( \frac{q}{\alpha}\right)\;,
 \end{align*}
 where $\tilde{q}_\mathcal{D}$ is the vector $\tilde{q}_\mathcal{D}=\argmax\limits_{\D\in\Delta_{k}} \left( \langle q, \mathcal{P}-\mathcal{D}\rangle+\alpha H(\mathcal{D}) \right)$. Then, 
 \begin{align*}
     \phi_\alpha(\mathcal{P}_\alpha^{\text{priv}})- \psi_\alpha(q_\alpha^{\text{out}})=& \langle \tilde{q}_{\mathcal{P}_\alpha^{\text{priv}}}, \mathcal{P}-\mathcal{P}_\alpha^{\text{priv}}\rangle+\alpha H(\mathcal{P}_\alpha^{\text{priv}})-\langle q_\alpha^{\text{out}}, \mathcal{P}\rangle+\alpha H^\ast \left( \frac{q_\alpha^{\text{out}}}{\alpha}\right)\\
     =&\langle \tilde{q}_{\mathcal{P}_\alpha^{\text{priv}}}, \mathcal{P}-\mathcal{P}_\alpha^{\text{priv}}\rangle-\langle q_\alpha^{\text{out}}, \mathcal{P}\rangle+ \langle \mathcal{P}_\alpha^{\text{priv}},q_\alpha^{\text{out}} \rangle\quad (\text{Corollary }\ref{KKT lemma})\\
  =&\langle \tilde{q}_{\mathcal{P}_\alpha^{\text{priv}}} - q_\alpha^{\text{out}}, \mathcal{P}-\mathcal{P}_\alpha^{\text{priv}}\rangle  \\
  =&\left\langle \tilde{q}_{\mathcal{P}_\alpha^{\text{priv}}} - q_\alpha^{\text{out}}, \mathcal{P} -\nabla H^\ast\left( \frac{q_\alpha^{\text{out}}}{\alpha}\right)\right\rangle  \quad \text{(Corollary \ref{KKT lemma})}\\
  =&\left\langle \tilde{q}_{\mathcal{P}_\alpha^{\text{priv}}} - q_\alpha^{\text{out}}, \nabla \psi_\alpha (q_\alpha^{\text{out}})\right\rangle \\
  \leq& \max_{q\in \text{conv}(\mathcal{Q})} \left\langle q - q_\alpha^{\text{out}}, \nabla \psi_\alpha (q_\alpha^{\text{out}})\right\rangle=:g_{\psi_\alpha}\left(q_\alpha^{\text{out}} \right)\;. 
\end{align*} 
\end{proof}
\begin{proof}[Proof of Theorem \ref{MAINRESULTS}]
A combination of the utility guarantee from Theorem \ref{thm:utility_guarantee} and the upper bound of the primal gap from Theorem \ref{thm:upper_bound_primalgap} yields 
\begin{align*}
      \EE \left[ \textbf{Gap}_{(P)}(\mathcal{P}_\alpha^{\text{priv}})\right]=\mathcal{O}\left( \frac{ \log^{1/4} (1/\delta) \log^{1/2}(|\mathcal{Q}|)}{\alpha^{1/4} (\varepsilon n)^{1/2}} + D_1\sqrt{\frac{ \log(k)}{n}}\right) +\alpha \log k\;.  
   \end{align*}
Optimizing the right hand side of the previous expression with respect to $\alpha$ yields the desired result. Finally, note that the $\textbf{Gap}_{(P)}(\mathcal{P}^{\text{priv}}_{\alpha^\ast})=\max\limits_{q\in \mathcal{Q}}\langle q,\mathcal{P}- \mathcal{P}^\text{priv}\rangle$ since $\phi^\ast =0\;.$ 
\end{proof}

\section{Proofs for the results in Section \ref{sec:primal}}\label{sec:proofs_5}
% \section{Randomized smoothing}
% \subsection{Missing proof at section 6.1 Randomized Smoothing}
This section introduces the concept of Rényi differential privacy, establishes related key technical results from optimization theory and provides proofs for  Theorems  \ref{thm:RS-upperbond1}
and \ref{thm:smoothing}

\subsection{The Gaussian Mechanism  and key properties of Rényi Differential Privacy}

 The  proof of Theorem
\ref{thm:smoothing}
 builds upon the definition of Rényi Differential Privacy, which we introduce here along with some of its key properties. We also describe the widely-used Gaussian mechanism, central to the privacy analysis of our algorithm. For a more comprehensive discussion of these topics, we refer the reader to \cite{Mironov2017RnyiDP} and \cite{dwork2014algorithmic}%, where most of the definitions can be found. 

\begin{defin}
For two probability distributions $f$ and $g$ supported over $\mathcal{R}$, the Rényi divergence of order $\beta>1$ is defined as 
\[ 
D_\beta(f\| g) :=\frac{1}{\beta-1} \log \mathrm{E}_{x \sim g}\left[\left(\frac{f(x)}{g(x)}\right)^\beta\right]. \]
\end{defin}

\begin{defin}
$((\beta, \epsilon)$-RDP). A randomized mechanism $f: \mathcal{Z}^n \mapsto$ $\mathcal{R}$ is said to be $\epsilon$-Rényi differentially private of order $\beta$, (denoted as $(\beta, \epsilon)$-RDP), if
\begin{equation}
\label{def:reny-dp}
    \sup_ { S_1 \sim S_2 } D_\beta\left(f(S_1) \| f\left(S_2\right)\right) \leq \epsilon \;.
\end{equation}
\end{defin}

Analogous to the Laplacian Mechanism, the so-called Gaussian Mechanism can be defined on the basis of a function 
$f:\mathcal{Z}^n\longrightarrow \RR^k$
and a dataset $S$
by 
\begin{equation}
    \label{eq:Gaussian_mechanism}
    \mathcal{A}(S)= f(S) + (Y_1, \ldots, Y_k)^\intercal\;,
\end{equation}
where $Y_j \overset{i.i.d}{\sim } \mathcal{N}\left(0,2\log(1.25/\delta)\Delta_2^2(f)/\varepsilon^2 \right)\;$, where the quantity 
        $\Delta_2(f) = \max_{S_1 \sim S_2} \| f(S_1)-f(S_2)\|_2\;$
        is called {\em $\ell^2-$sensitivity of $f$}, and where the supremum is taken over all neighboring datasets.

For any choice of $\varepsilon$ and $\delta$, the Gaussian Mechanism satisfies $(\varepsilon, \delta)$-differential privacy (see Appendix A of \cite{dwork2014algorithmic}). Moreover, it can be shown  that the Gaussian Mechanism achieves privacy in the sense of Rényi Differential Privacy (RDP). More specifically, for any $\beta>1$ and $\sigma>0$, 
\begin{equation}
    \label{eq:Gaussian_mechanism_Renyi}
    \mathcal{A}(S)= f(S) + (Y_1, \ldots, Y_k)^\intercal\;, \quad Y_j\overset{i.i.d}{\sim}\mathcal{N}(0, \sigma^2)\;,
\end{equation}
is  $\left(\beta, \frac{1}{\sigma^2}\frac{\beta}{2} \Delta_2^2(f)\right)$-RDP. In order to prove \eqref{eq:Gaussian_mechanism_Renyi}, we will utilize the following lemma, which outlines a well-known property of Gaussian distributions.
\begin{lema}\label{lem:Gaussian_renyi}
  For $\mu_1, \mu_2 \in \RR^k$ and $\Sigma $ a positive definite matrix of size $k$, let $f \sim \mathcal{N}(\mu_1, \Sigma)$ and $g \sim \mathcal{N}(\mu_2, \Sigma) $. Then, for any $\beta>0$ the following holds  
  \begin{align*}
    D_\beta ( f \, ||\, g)= \frac{\beta}{2} (\mu_1-\mu_2)^\intercal \Sigma^{-1} (\mu_1 - \mu_2)\;.  
  \end{align*}
\end{lema}
\begin{proof}

For \( f\sim \mathcal{N}(\mu_1, \Sigma) \) and \( g\sim \mathcal{N}(\mu_2, \Sigma) \) it holds that
\begin{align*}
\frac{f(x)}{g(x)} =& \exp{ \left( -\frac{1}{2} \left( (x - \mu_1)^\intercal \Sigma^{-1} (x - \mu_1) - (x - \mu_2)^\intercal \Sigma^{-1} (x - \mu_2) \right) \right)} \\
=& \exp \left( -\frac{1}{2} \left( x^\intercal \Sigma^{-1} x - 2 \mu_1^\intercal \Sigma^{-1} x + \mu_1^\intercal \Sigma^{-1} \mu_1 - x^\intercal \Sigma^{-1} x + 2 \mu_2^\intercal \Sigma^{-1} x - \mu_2^\intercal \Sigma^{-1} \mu_2 \right) \right)\\
 =& \exp{\left( (\mu_1 - \mu_2)^\intercal \Sigma^{-1} x + \frac{1}{2} \left( \mu_2^\intercal \Sigma^{-1} \mu_2  - \mu_1^\intercal \Sigma^{-1} \mu_1 \right) \right)}.
\end{align*}
From this, it follows that
\begin{align*}
\mathbb{E}_{x \sim g} \left( \frac{f(x)}{g(x)} \right)^\beta =& \mathbb{E}_{x \sim \mathcal{N}(\mu_2, \Sigma)} \exp \left( \beta \left( (\mu_1 - \mu_2)^\intercal \Sigma^{-1} x + \frac{1}{2} (\mu_2^\intercal \Sigma^{-1} \mu_2 - \mu_1^\intercal \Sigma^{-1} \mu_1 ) \right) \right)\\
=&\mathbb{E}_{x \sim \mathcal{N}(\mu_2, \Sigma)} \exp \left( \beta (\mu_1 - \mu_2)^\intercal \Sigma^{-1} x \right) \exp \left( \frac{\beta}{2} (\mu_2^\intercal \Sigma^{-1} \mu_2 - \mu_1^\intercal \Sigma^{-1} \mu_1 ) \right).
\end{align*}
The second factor is a constant and can be taken out of the expectation, we focus on computing the expectation:
\[
\mathbb{E}_{x \sim \mathcal{N}(\mu_2, \Sigma)} \exp \left( \beta (\mu_1 - \mu_2)^\intercal \Sigma^{-1} x \right).
\]
Since \( x \sim \mathcal{N}(\mu_2, \Sigma) \), we know that \( x = \mu_2 + z \), where \( z \sim \mathcal{N}(0, \Sigma) \). Substituting this into the expectation:
\[
\mathbb{E}_{z \sim \mathcal{N}(0, \Sigma)} \exp \left( \beta (\mu_1 - \mu_2)^\intercal \Sigma^{-1} (\mu_2 + z) \right)=
\exp \left( \beta (\mu_1 - \mu_2)^\intercal \Sigma^{-1} \mu_2 \right) \mathbb{E}_{z \sim \mathcal{N}(0, \Sigma)} \exp \left( \beta (\mu_1 - \mu_2)^\intercal \Sigma^{-1} z \right).
\]
The expectation \( \mathbb{E}_{z \sim \mathcal{N}(0, \Sigma)} \exp(a^\intercal z) \) for a normal random variable \( z \sim \mathcal{N}(0, \Sigma) \) is given by \( \exp\left( \frac{1}{2} a^\intercal \Sigma a \right) \), where \( a \) is a vector. Thus, we have
\[
\mathbb{E}_{z \sim \mathcal{N}(0, \Sigma)} \exp \left( \beta (\mu_1 - \mu_2)^\intercal \Sigma^{-1} z \right) = \exp \left( \frac{\beta^2}{2} (\mu_2 - \mu_1)^\intercal \Sigma^{-1} (\mu_2 - \mu_1) \right).
\]

Combining everything, we get
\[
\mathbb{E}_{x \sim \mathcal{N}(\mu_2, \Sigma)} \left( \frac{f(x)}{g(x)} \right)^\beta = \exp \left( \beta (\mu_1 - \mu_2)^\intercal \Sigma^{-1} \mu_2 + \frac{\beta^2}{2} (\mu_2 - \mu_1)^\intercal \Sigma^{-1} (\mu_2 - \mu_1) + \frac{\beta}{2} (\mu_2^\intercal \Sigma^{-1} \mu_2 - \mu_1^\intercal \Sigma^{-1} \mu_1) \right).
\]
The first term and the last term simplify to
\begin{align*}
& \exp \left( \beta (\mu_1 - \mu_2)^\intercal \Sigma^{-1} \mu_2  + \frac{\beta}{2} (\mu_2^\intercal \Sigma^{-1} \mu_2 - \mu_1^\intercal \Sigma^{-1} \mu_1) \right) \\
= & \exp \left( \frac{\beta}{2} \Big[ (\mu_1 - \mu_2)^\intercal \Sigma^{-1} \mu_2  + \mu_1 ^\intercal \Sigma^{-1} \mu_2 -  \mu_2 ^\intercal \Sigma^{-1}  \mu_2 + \mu_2^\intercal \Sigma^{-1} \mu_2 - \mu_1^\intercal \Sigma^{-1} \mu_1 \Big] \right) \\
= & \exp \left( \frac{\beta}{2} \Big[ (\mu_1 - \mu_2)^\intercal \Sigma^{-1} \mu_2  + \mu_1 ^\intercal \Sigma^{-1} (\mu_2 - \mu_1 )\Big] \right) \\
= & \exp \left(- \frac{\beta}{2} \Big[ (\mu_2 - \mu_1)^\intercal \Sigma^{-1} (\mu_2 - \mu_1 )\Big] \right).
\end{align*}

In the end, we obtain
\[
\mathbb{E}_{x \sim \mathcal{N}(\mu_2, \Sigma)} \left( \frac{f(x)}{g(x)} \right)^\beta  = \frac{\beta(\beta-1)}{2} (\mu_2 - \mu_1)^\intercal \Sigma^{-1} (\mu_2 - \mu_1).
\]
Thus, the Rényi divergence is
\[
D_\beta(f \| g) = \frac{\beta}{2} (\mu_2 - \mu_1)^\intercal \Sigma^{-1} (\mu_2 - \mu_1).
\]
\end{proof}

 Lemma \ref{lem:Gaussian_renyi} shows a way of ensuring RDP for \eqref{eq:Gaussian_mechanism_Renyi}. In fact, notice that 
$$ \mathcal{A}(S) \sim \mathcal{N}\left( f(S), \sigma^2 I_{k} \right)\;.$$ Then, by  Lemma \ref{lem:Gaussian_renyi}, it follows that 
\begin{align}
\label{eq:discussion_gaussian_privacy}
     \sup_ { S_1\sim S_2 } D_\beta\left(f(S_1) \| f(S_2)\right)=&   \sup_ { S_1\sim S_2 } \frac{1}{\sigma^2}\frac{\beta}{2}   \sup_ { S_1\sim S_2 }\| f(S_1)- f(S_2)\|_2^2=  \frac{1}{\sigma^2}\frac{\beta}{2} \Delta_2^2(f)\;.\\
\end{align}
The previous expression proves that the Gaussian mechanism defined by \eqref{eq:Gaussian_mechanism_Renyi} is $\left(\beta, \frac{1}{\sigma^2}\frac{\beta}{2} \Delta_2^2(f)\right)-$RDP.

We conclude this part by reporting three major properties of RDP (\cite{Mironov2017RnyiDP}):
\begin{enumerate}
    \item \textbf{Post-processing}: If $\mathcal{A}$ is $(\beta, \epsilon)$-RDP and $g$ is a randomized mapping, then $g \circ A$ is $(\beta, \epsilon)$-RDP.\\
\item \textbf{ Adaptive Composition}: If $\mathcal{A}_1 : \mathcal{Z}^n \to \mathcal{X}_1$ is $(\beta, \epsilon_1)$-RDP, and $\mathcal{A}_2 : \mathcal{X}_1 \times \mathcal{Z}^n \to \mathcal{X}_2$ is $(\beta, \epsilon_2)$-RDP, then
\[
(\mathcal{A}_1(\cdot), \mathcal{A}_2(\mathcal{A}_1(\cdot), \cdot)) \text{ is } (\beta, \epsilon_1 + \epsilon_2)\text{-RDP.}
\]
\item \textbf{ From RDP to DP}: If $\mathcal{A}$ is $(\beta, \epsilon)$-RDP, then it is $\left( \epsilon + \frac{\log(1/\delta)}{\beta - 1}, \delta \right)$-DP for any $0 < \delta < 1$.
\end{enumerate}

\subsection{Proof of Theorems   \ref{thm:RS-upperbond1}
and \ref{thm:smoothing}}
\label{Proof:smoothing}
In this section we present the proof of privacy and convergence for Algorithm \ref{Algo ACSMD}. The privacy is injected into the algorithm through the randomized smoothing of $\phi$. The overall analysis of  Algorithm \ref{Algo ACSMD} then follows from an application of the Advanced Composition Theorem \ref{thm:advance_composition_theorem}. The only difference with the privacy analysis of the Frank-Wolfe Algorithm (see Theorem \ref{thm:utility_guarantee}) is that we use the closely related definition of  R\'enyi-differential privacy (RDP) (see \cite{Mironov2017RnyiDP}). To motivate this choice, the randomized smoothing we present is based on the introduction of Gaussian noise, which turns out to be easier to analyse in terms of RDP. Finally, the privacy guarantee will be translated in terms of the classical $(\varepsilon, \delta)$ definition of privacy. The accuracy guarantee will be a direct consequence of Theorem 3.4 in \cite{daspremont:hal-04230893}.

\label{appendix-stochastic-oracle}
More precisely,  in order to use the setting presented in \cite{daspremont:hal-04230893},   we need to verify the following properties of the smoothed gradient:
\begin{itemize}
\item $\phi_{\sigma}:\Delta_{k}\mapsto\RR$ is a convex, 1-Lipschitz and $1/\sigma$-smooth function with respect to $\|\cdot\|_1$. Given a Gaussian sampler $\xi\sim{\cal N}(0,\sigma^2I)$, we have access to a stochastic first-order  oracle for $\phi_{\sigma}$, namely
\[ G(\mathcal{D},\xi, S_n, \mathcal{Q})  \in\argmax_{ q \in \text{conv}(\mathcal{Q})} \langle q ,{\mathcal{P}} - \mathcal{D}+\xi\rangle. \]
Notice that almost surely this is a well-defined and unique assignment. 
In fact, the maximum of the  linear function $q\mapsto \langle q ,{\mathcal{P}} - \mathcal{D}\rangle$ over the  polyhedron $\text{conv}(\mathcal{Q})$ takes its maximum in at least one of its finitely many vertices. The maximizer is not unique only if one of the vertices is orthogonal to ${\mathcal{P}} - \mathcal{D}+\xi.$ As $\xi$ has a continuous distribution on $\mathbb{R}^d$, the maximizer is almost surely unique.
Further,
\begin{align}
&\mathbb{E}_{\xi}[G(\mathcal{D},\xi, S_n, \mathcal{Q}) ]=\nabla \phi_{\sigma}(\mathcal{D})\\   
&\mathbb{E}_{\xi}\|G(\mathcal{D},\xi, S_n, \mathcal{Q}) -\nabla \phi_{\sigma}( \mathcal{D})\|_{\infty}^2 \leq 2, \label{eq:bbd_variance_phi_sigma}
\end{align}
where the first equality holds by the Dominated Convergence Theorem, and the inequality holds since the $\ell_{\infty}$-diameter of $Q$ is at most 2.
    \item $H:\Delta_{k}\mapsto\RR_+$ is $1$-strongly convex with respect to $\|\cdot\|_1$ (see Example 2.5 in \cite{shalev2012online}).
\end{itemize}

\begin{proof}[Proof of Theorem \ref{thm:smoothing}]

Proposition \ref{prop:rand_smooth} guarantees $\phi_\sigma$ to be  $1/\sigma$-smooth, and from \eqref{eq:bbd_variance_phi_sigma} it follows that the variance of the estimator $\Phi(\mathcal{D},\xi)$ is bounded. Then, according to Theorem 3.4 in \cite{daspremont:hal-04230893}, the accuracy guarantee for $(P_{\alpha,\sigma})$ achieves the rate
\begin{equation}
\label{eq-smooth-1}
\EE\left[ \textbf{Gap}_{(P_{\alpha,\sigma})}(\mathcal{P}^{\text{priv}}_\alpha)  \right]  \leq  \mathcal{O}
\left( \frac{1}{\alpha \sigma T^2 } + \frac{1}{\alpha T} \right).
\end{equation}

\textbf{Choice of $\sigma$:} The parameter $\sigma$ corresponds to the variance of the Gaussian noise in the randomized smoothing, and it regulates the privacy budget of each iteration of Algorithm \ref{Algo ACSMD}. Therefore, the choice of $\sigma$ follows from the overall privacy of the Algorithm. 

The privacy enters the objective through $G_\sigma(\mathcal{D}_t^{md}, \xi_t, S_n, \mathcal{Q})=\argmax_{q\in \text{conv}(\mathcal{Q})} \langle q, \mathcal{P}_n-\mathcal{D}+\xi \rangle\;,$  which is the composition of a Gaussian mechanism and a deterministic function. More precisely, let $f$ be the $\RR^k$-valued function of the dataset $S_n$ defined as $ f(S_n):= \mathcal{P}_n -\mathcal{D}$. We define the Gaussian mechanism as $\mathcal{M}(S_n)= f(S_n)+ \xi$. The sensitivity of $f$ is then given by 
\begin{align*}
    \Delta_2(f)=& \sup_{S_n \sim S_n'} \| f(S_n)- f(S_n')\|_2 = \sup_{S_n \sim S_n'} \left\| \mathcal{P}_n - \mathcal{D} -\mathcal{P}_n' + \mathcal{D}\right\|_2 \\
    =& \sup_{S_n \sim S_n'} \left\| \frac{1}{n}\sum_{i=1}^n e_{z_i} - \frac{1}{n}\sum_{i=1}^n e_{z'_i} \right\|_2  = \frac{\sqrt{2}}{n}\;.\\
\end{align*}
The last step is justified by the fact that if $S_n$ and $S_n'$ are neighboring, then there exists only one index $i$ for which $z_i\neq z_i'$, and the corresponding $ e_{z_i}$ and $ e_{z'_i}$ have distance $\sqrt{2}$ with respect to the $\|\cdot\|_2$-norm.  Hence, an application of the Gaussian mechanism yields $\mathcal{M}$ to be $\left(\beta, \frac{1}{\sigma^2} \frac{\beta}{ n^2}\right) $-RDP (see \eqref{eq:discussion_gaussian_privacy}) . Notice the conclusion remains the same for $G_\sigma $ by the post-processing property of RDP. 

For $T$ iterations, the Advanced Composition Theorem for RDP ensures that our algorithm is $(\beta,\frac{\beta T}{n^2\sigma^2})$-RDP. Lastly, by transfering privacy from RDP to DP, the whole Algorithm is $\big(\frac{\beta T}{n^2\sigma^2}+\frac{\log 1/\delta}{\beta-1},\delta \big)$-DP. Optimal tuning with respect to $\beta>1$ (i.e.  minimizing the privacy budget $\frac{\beta T}{n^2\sigma^2}+\frac{\log 1/\delta}{\beta-1}$ over $\beta$), leads to
    \[ \beta^\ast=1+\sqrt{\frac{\log1/\delta}{T}}n\sigma\;. \]
Finally, for a chosen $\varepsilon>0$, we  obtain $\frac{\beta^\ast T}{(n\sigma)^2}+\frac{\log 1/\delta}{\beta^\ast-1}\leq \varepsilon$ by setting
    \( \sigma=\frac{4\sqrt{T\log1/\delta}}{n\varepsilon}\;. \) In fact, 
    \begin{align*}
        \varepsilon \geq \frac{\beta^\ast T}{(n\sigma)^2}+\frac{\log (1/\delta)}{\beta^\ast-1} = 2 \sqrt{\frac{T \log (1/\delta)}{n^2 \sigma^2}} + \frac{T}{n^2\sigma^2}\;.
    \end{align*}
By multiplying on both sides of the previous expression with $\sigma^2>0$, we obtain 
\begin{align*}
     \varepsilon \geq \frac{\beta^\ast T}{(n\sigma)^2}+\frac{\log (1/\delta)}{\beta^\ast-1} \Longleftrightarrow \varepsilon \sigma^2 -2 \sigma \sqrt{\frac{T\log (1/\delta)}{n^2}} - \frac{T}{n^2}\geq 0\;.
\end{align*}
The two roots of the corresponding polynomial of second degree in $\sigma$ are 
\begin{align*}
    \sigma_1 := \frac{1}{\varepsilon n }\left( \sqrt{{T \log(1/\delta)}} - \sqrt{ T\log(1/\delta) +\varepsilon T} \right)\;,\quad  \sigma_2:=\frac{1}{\varepsilon n }\left( \sqrt{{T \log(1/\delta)} }+ \sqrt{ T\log(1/\delta) +\varepsilon T} \right)\;,
\end{align*}
and any $\sigma \in (-\infty, \sigma_1]\cup [\sigma_2, +\infty)$ satisfies the inequality. Finally, note that for $\delta$ sufficiently small, $\log(1/\delta)>\varepsilon$, thus 
\begin{align*}
    \sigma_2 \leq \frac{(1+\sqrt{2}) \sqrt{T \log (1/\delta)}}{\varepsilon n} \leq   \frac{ 4 \sqrt{T \log (1/\delta)}}{\varepsilon n} \;.
\end{align*}
\textbf{Choice of $\alpha$:} The choice of $\alpha$ is based on the accuracy gap between  $(P_{\alpha,\sigma})$ and $(P)$. In the proof of Theorem \ref{thm:upper_bound_primalgap} we derived the upper bound 
\begin{align*}
    \textbf{Gap}_{(P)}(\mathcal{P}_\alpha^{\text{priv}}) = & \phi(\mathcal{P}_\alpha^{\text{priv}})-\phi^\ast \leq \phi_\alpha(\mathcal{P}_\alpha^{\text{priv}}) -\phi^\ast_\alpha +\alpha \log k = \textbf{Gap}_{(P_{\alpha})}(\mathcal{P}_\alpha^{\text{priv}}) + \alpha \log k\;.\\
\end{align*}
Taking the expectation on both sides yields
\begin{align*}
\EE[\textbf{Gap}_{(P)}(\mathcal{P}^{\text{priv}}_\alpha) ] & \leq  \EE[\textbf{Gap}_{(P_{\alpha})}(\mathcal{P}^{\text{priv}}_\alpha)] +\alpha \log k\;.
\end{align*}
Further, for any $\mathcal{D} \in \Delta_k $
\begin{align*}
\sigma w(\mathcal{Q}) & :=  \mathbb{E}_{\xi \sim \varphi_\sigma}\Big[\max_{q\in Q}\langle q, \xi\rangle \Big] \\
& \geq \mathbb{E}_{\xi\sim \varphi_\sigma}\Big[\max_{q\in Q}\langle q, {\mathcal{P}} -\mathcal{D}+\xi\rangle\Big]-\max_{q\in \mathcal{Q}}\langle q,\mathcal{P}-\mathcal{D} \rangle \\
& = \phi_\sigma(\mathcal{D}) - \phi(\mathcal{D})
\end{align*}
Moreover, due to symmetry of $\xi$, it holds that
\begin{align*}
\sigma w(\mathcal{Q}) & \geq \max_{q\in \mathcal{Q}}\langle q,\mathcal{P}-\mathcal{D} \rangle - \mathbb{E}_{\xi\sim \varphi_\sigma}\Big[\max_{q\in Q}\langle q, {\mathcal{P}} -\mathcal{D} - \xi\rangle\Big] \\
& = \phi(\mathcal{D}) - \phi_\sigma(\mathcal{D}).
\end{align*}

Then, for $\D_\alpha^\ast=\argmin\limits_{ \mathcal{D} \in \Delta_{k}} \left( \langle q_\alpha^\ast, \mathcal{P} - \mathcal{D}\rangle+\alpha H(\mathcal{D}) \right)$ ,
\begin{align*}
\EE[\textbf{Gap}_{(P_{\alpha})}(\mathcal{P}^\text{priv}_\alpha) ]  =&
\Big[ \phi(\mathcal{P}^\text{priv}_\alpha) + \alpha H(\mathcal{P}^\text{priv}_\alpha) \Big] - \Big[  \phi(\mathcal{D}^\ast_{\alpha}) + \alpha H(\mathcal{D}^\ast_{\alpha}) \Big] \\
 =&\Big[ \phi_\sigma(\mathcal{P}^\text{priv}_\alpha) + \alpha H(\mathcal{P}^\text{priv}) \Big] - \Big[  \phi_\sigma (\mathcal{D}^\ast_{\alpha}) + \alpha H(\mathcal{D}^\ast_{\alpha}) \Big] \\
&  + \Big( \phi(\mathcal{P}^\text{priv}_\alpha) - \phi_\sigma(\mathcal{P}^\text{priv}_\alpha) \Big)+ \Big( \phi_\sigma (\mathcal{D}^\ast_{\alpha}) - \phi(\mathcal{D}^\ast_{\alpha}) \Big) \\
  \leq & \EE\left[ \textbf{Gap}_{(P_{\alpha,\sigma})}(\mathcal{P}^\text{priv}_\alpha) \right] + 2 \sigma w(\mathcal{Q})\;.
\end{align*}

Together with \eqref{eq-smooth-1} this yields
\begin{align*}
\EE[\textbf{Gap}_{(P_{\alpha})}(\mathcal{P}^\text{priv}_\alpha) ] & \leq \EE\left[ \textbf{Gap}_{(P_{\alpha,\sigma})}(\mathcal{P}^\text{priv}_\alpha) \right]  + 2 \sigma w(\mathcal{Q})\;
=  \mathcal{O}
\left( \frac{1}{\alpha \sigma T^2 } + \frac{1}{\alpha T} \right)+ 2 \sigma w(\mathcal{Q})\;.
\end{align*}

Note that in the upper bound from \eqref{eq-smooth-1}, the term $O(\frac{1}{\alpha \sigma T^2})$ can be reduced in practise with other techniques (such as restarting algorithm in \cite{daspremont:hal-04230893}) such that $\mathcal{O}(\frac{1}{\alpha T})$ becomes dominant while to the best of our knowledge, there is no technique to reduce $\mathcal{O}(\frac{1}{\alpha T})$. For simplicity, in our analysis, we will focus on the case where  $\mathcal{O}(\frac{1}{\alpha T})$ is dominant compared to the other term i.e. \( \frac{1}{\sigma} =\mathcal{O}(T)\). Later, we will show that this regime is justified if $n$ is sufficiently large.

We will show that the optimal number of iterations, denoted as \(T^\ast\) satisfies this inequality and is consistent with the value stated in the theorem, i.e we will have to verify that  $ \frac{\sqrt{T\log(1/\delta)}}{\varepsilon n} \geq \frac{1}{T}$ (by rearranging the inequality and plugging $\sigma=\frac{ 4 \sqrt{T \log (1/\delta)}}{\varepsilon n} $). 

Therefore, we obtain
\begin{align*}
\EE[\textbf{Gap}_{(P_{\alpha})}(\mathcal{P}^\text{priv}_\alpha) ] \leq   \mathcal{O}
\left(  \frac{1}{\alpha T} \right)+ 2 \sigma w(\mathcal{Q})\;.
\end{align*}

By combining all the previous inequalities and substituting the value of \( \sigma \) obtained from the privacy analysis, we arrive at

\begin{align*}
\EE[\textbf{Gap}_{(P)}(\mathcal{P}^{\text{priv}}_\alpha) ] & \leq  \EE[\textbf{Gap}_{(P_{\alpha})}(\mathcal{P}^{\text{priv}}_\alpha)] +\alpha \log k \leq \mathcal{O}
\left(  \frac{1}{\alpha T} +\alpha \log k  +  \frac{\sqrt{T\log(1/\delta)}}{\varepsilon n} \ w(\mathcal{Q}) \right)\;. 
\end{align*}

Next, we optimize the right hand side of the upper bound above, obtaining over $\alpha$, yielding to $\alpha^\ast= \frac{1}{\sqrt{T \log k}}$ and 

\begin{align*}
\EE[\textbf{Gap}_{(P)}(\mathcal{P}^{\text{priv}}_\alpha) ]\leq \mathcal{O}
\left(  \frac{\sqrt{\log k}}{\sqrt{T}}  +  \frac{\sqrt{T\log(1/\delta)}}{\varepsilon n} \ w(\mathcal{Q}) \right)\;. 
\end{align*}
Finally, a last minimization over $T$ yields
\[ T^{\ast}=\sqrt{\frac{\log k}{\log(1/\delta)}}\frac{n\varepsilon}{w(\mathcal{Q})}\;,\]
and 
\[ \mathcal{O} \left( \frac{w(\mathcal{Q})^2 \sqrt{\log(1/\delta)} }{n \epsilon \sqrt{\log k}}+ \frac{\sqrt{w(\mathcal{Q})}\log^{1/4} (k)  \log^{1/4}(1/\delta)}{n^{1/2}\varepsilon^{1/2}} \right)\;. \]

The second term of the previous expression is dominant when $n\to \infty\;.$

Finally, note that the condition $T^{\ast}\geq \frac{1}{\sigma}$ is satisfied whenever $n \geq \frac{w^3(\mathcal{Q}) \log(1/\delta)}{\epsilon \log^{3/2}(k)}$, which is one of the assumptions of Theorem \ref{thm:smoothing}. 
\end{proof}

\section{Private lower bounds}\label{sec:lower_bounds}

In this section, we provide evidence for the efficiency of our rates by studying lower bounds for  private synthetic data release. We note that previous work has established such lower bounds (see, e.g., \cite{Ullman:2011,Bun:2018, Steinke:2016}), and, in fact, we will make use of some of these results. The population private lower bound is based on a lower bound previously established for the empirical DP synthetic data problem. We summarize this result below:

\begin{teo} \label{thm-bun}
[Lemma 2.11, Theorem 5.16 in \cite{Bun:2018}] For all $\gamma>0$, $\log(k)\geq 6\log(1/\gamma)$, and $k\geq \log(k)/\gamma^2$, there exists a family of queries $\mathcal{Q}$ and a support $\mathcal{Z}$ of size $k$ such that any $(\varepsilon,\delta)$-DP algorithm that takes as input a dataset $S_n$, and outputs ${\cal A}(S_n)$ such that
\( \mathbb{E}_{{\cal A}}\Big[ \max_{q\in \mbox{conv} ({\cal Q})} \langle q,\mathcal{P}_n-{\cal A}(S_n)\rangle\Big] \leq \gamma\;,\)
requires at least $n=\tilde\Omega\left(\frac{\sqrt{\log(k)}\log|\mathcal{Q}|}{\varepsilon\gamma^2} \right) \;.$

\end{teo}

In particular, 
Theorem \ref{thm-bun} implies  that for any \((\varepsilon, \delta)\)-DP algorithm \(\mathcal{A}\) that satisfies the assumptions of Theorem \ref{thm-bun} and outputs a synthetic data  distribution $\mathcal{A}(S_n)$ with \( \mathbb{E}_{{\cal A}}\Big[ \max_{q\in \mbox{conv} ({\cal Q})} \langle q,\mathcal{P}_n-{\cal A}(S_n)\rangle\Big] \leq \gamma\;,\) it must hold that 
\[
\gamma \geq \tilde\Omega\left(\log^{1/4} (k)\sqrt{\frac{\log|\mathcal{Q}|}{\varepsilon n}} \right).
\]

 This result establishes a fundamental lower bound on the accuracy \(\gamma\) for such algorithms, indicating that \(\gamma\) cannot be arbitrarily small without a corresponding increase in the dataset size \(n\).

%with $\mathcal{P}_n$ the empirical distribution of $S_n$
Next we provide a reduction from empirical to population guarantees for DP algorithms. This reduction is based on a known resampling argument used in DP stochastic convex optimization \cite{BFTT19} and stochastic variational inequalities \cite{Boob2021OptimalAF}.

\begin{lema} [Appendix C in \cite{BFTT19}] \label{lower-lemma}
Let \(\mathcal{P}\) be a distribution supported on \(\mathcal{Z}\), and let \(\mathcal{P}_n\) denote the empirical distribution corresponding to a dataset \(S_n \in \mathcal{Z}^n\). Suppose there exists an \((\varepsilon/[4\log(1/\delta)], e^{-\varepsilon}\delta/[8\log(1/\delta)])\)-differentially private (DP) algorithm \(\mathcal{A}\) such that

\[
\mathbb{E}_{\mathcal{A}} \left[ \max_{q \in \text{conv}(\mathcal{Q})} \langle q, \mathcal{P} - \mathcal{A}(S_n) \rangle \right] \leq \gamma,
\]
where \(\gamma > 0\) is a fixed accuracy parameter.
Then, there exists an \((\varepsilon, \delta)\)-DP algorithm \(\mathcal{B}\) such that

\[
\mathbb{E}_{\mathcal{B}} \left[ \max_{q \in \text{conv}(\mathcal{Q})} \langle q, \mathcal{P}_n - \mathcal{B}(S_n) \rangle \right] \leq \gamma.
\]

\end{lema}

\begin{proof}
Given algorithm ${\cal A}$ as in the statement and dataset $S_n\in\mathcal{Z}^n$, consider the following algorithm ${\cal B}$, which takes $S_n$ as input, and does the following: first, it samples $n$ independent copies from the empirical distribution associated to $S_n$. Calling this new dataset $T_n$, it then runs ${\cal A}$ on $T_n$ and outputs ${\cal B}(S_n):={\cal A}(T_n)$.

First, ${\cal B}$ is $(\varepsilon,\delta)$-DP with respect to  $S_n$: this follows from a simple analysis of the probability of repeating examples from $S_n$ in dataset $T_n$, together with the group privacy property (see Appendix C in \cite{BFTT19}). Next,
\begin{align*}
& \mathbb{E}_{\cal B}\big[\max_{q\in \mbox{conv} ({\cal Q})} \langle q,\mathcal{P}_n-{\cal B}(S_n)\rangle \big] \\
= & \mathbb{E}_{{\cal A},T_n \sim (\mathcal{P}_n)^n}\big[\max_{q\in \mbox{conv} ({\cal Q})}\langle q,\mathcal{P}_n-{\cal A}(T_n)\rangle\big] \leq \gamma,
\end{align*}
where the last inequality holds by the population accuracy of ${\cal A}$.
\end{proof}

Combining Theorem \ref{thm-bun} and Lemma \ref{lower-lemma}, we obtain the following lower bound for population loss.

\begin{teo}\label{thm:lower_bound}
For all $\gamma>0$, $\log(k)\geq 6\log(1/\gamma)$, and $k\geq \log(k)/\gamma^2$, there exists a family of queries $\mathcal{Q}$ and a support $\mathcal{Z}$ of size $k$ such that any $(\varepsilon, \delta)$-DP algorithm that takes as input a dataset $S_n \sim {\mathcal{P}}^n$, and outputs ${\cal A}(S_n)$ such that
%$(\varepsilon/[4\log(1/\delta)],e^{-\varepsilon}\delta/[8\log(1/\delta)])$-DP algorithm that takes as input a dataset $S_n \sim {\mathcal{P}}^n$, and outputs ${\cal A}(S_n)$ such that
\[ \mathbb{E}_{{\cal A}, S_n}\Big[ \max_{q\in \mbox{conv} ({\cal Q})} \langle q,\mathcal{P}-{\cal A}(S_n)\rangle\Big] \leq \gamma\;, \]
satisfies
\[ \gamma=\tilde\Omega\left(\frac{\log^{1/4}( k)\log^{1/2}|\mathcal{Q}|}{\varepsilon^{1/2} n^{1/2}} \right). \]
\end{teo}

Let $\mathcal{Q}$ denote the query class derived from Theorem \ref{thm:lower_bound}, and let $\mathcal{A}(\cdot)$ represent Algorithm \ref{Algo ACSMD}, which satisfies $(\varepsilon, \delta)$-differential privacy. For a fixed $\delta > 0$ and $\mathcal{P}^{\text{priv}} = \mathcal{A}(S_n)$,  the theorem gives
\[  \EE\left[\max\limits_{q\in \mathcal{Q}} \, \langle q, {\mathcal{P}} - \mathcal{P}^{\text{priv}} \rangle\right] =\tilde\Omega\left(\frac{\log^{1/4}( k)\log^{1/2}|\mathcal{Q}|}{\varepsilon^{1/2} n^{1/2}} \right).  \]

We observe that this lower bound differs by a factor of approximately $w^{1/2}(\mathcal{Q}) \log^{1/4}(1/\delta)$ compared to the upper bound established by Algorithm \ref{Algo ACSMD} in Theorem \ref{thm:smoothing}, which is
 \begin{align*}
 \EE\left[\max\limits_{q\in \mathcal{Q}} \, \langle q, {\mathcal{P}} - \mathcal{P}^{\text{priv}} \rangle\right] = \mathcal{O}  \left( \frac{w^{1/2}(\mathcal{Q}) \log^{1/4}(k)\log^{1/4}( 1/\delta)}{ \varepsilon^{1/2} n^{1/2}} \right)\;.
 \end{align*}

In particular, since the lower bound does not explicitly depend on \(\delta\), our upper bound on the utility guarantee is optimal when \(\mathcal{Q} \subset \mathcal{B}^k_2\). As mentioned towards the end of the paper, in this case, \(w(\mathcal{Q}) = \mathcal{O}(\log |\mathcal{Q}|)\), which aligns with the upper bound established in Theorem \ref{thm:lower_bound} for a fix $\delta>0\;.$

\end{document}